\documentclass[11pt,oneside]{amsart}

\usepackage{hyperref}
\usepackage[a4paper,margin=3cm]{geometry}
\usepackage{amssymb}
\usepackage[foot]{amsaddr}
\usepackage[T1]{fontenc}
\usepackage[utf8]{inputenc}
\usepackage{enumerate}
\usepackage{todonotes}
\usepackage{microtype}
\usepackage{marvosym}
\usepackage{mathrsfs}
\usepackage{tikz}
\usetikzlibrary{calc}

\usepackage{thmtools, thm-restate}
\declaretheorem[name=Theorem]{thm}
\declaretheorem[name=Proposition,sibling=thm]{prop}
\declaretheorem[name=Lemma,sibling=thm]{lem}
\declaretheorem[name=Conjecture,sibling=thm]{conj}
\declaretheorem[name=Definition,sibling=thm,style=definition]{defn}
\declaretheorem[name=Corollary,sibling=thm]{cor}
\declaretheorem[name=Observation,sibling=thm]{obs}

\declaretheorem[name=Problem]{prob}
\declaretheorem[name=Assumption]{ass}
\declaretheorem[name=Consequence,sibling=ass]{consq}

\renewcommand{\int}{\ensuremath{\textup{\textsf{int}}}}

\newcommand{\conv}{\operatorname{conv}}
\newcommand{\xc}{\operatorname{xc}}

\newcommand{\R}{\mathbb{R}}
\newcommand{\Z}{\mathbb{Z}}
\newcommand{\Q}{\mathbb{Q}}
\newcommand{\N}{\mathbb{N}}

\newcommand{\ocp}{\operatorname{ocp}}
\newcommand{\STAB}{\operatorname{STAB}}

\newcommand{\stab}{\STAB}
\newcommand{\zerovec}{\mathbf{0}}
\newcommand{\onevec}{\mathbf{1}}

\newcommand{\EP}{Erd\H{o}s-P\'osa}
\newcommand{\surf}{\mathbb{S}}
\newcommand{\cut}{\text{\LeftScissors}}

\newcommand{\maxdet}{\Delta}
\newcommand{\sub}{P}
\newcommand{\slack}{Q}

\renewcommand{\ge}{\geqslant}
\renewcommand{\le}{\leqslant}

\title[Stable sets in bounded genus, bounded OCP graphs]{The stable set problem in graphs with bounded genus and bounded odd cycle packing number}

\author{Michele Conforti\textsuperscript{1}} 
\author{Samuel Fiorini\textsuperscript{2}}
\author{Tony Huynh\textsuperscript{3}} 
\author{Gwena\"el Joret\textsuperscript{4}} 
\author{Stefan Weltge\textsuperscript{5}}
\address[1]{Universit\`a degli Studi di Padova}
\address[2,3,4]{Universit\'{e} libre de Bruxelles}
\address[5]{Technische Universit\"at M\"unchen}

\begin{document}

\begin{abstract}
Consider the family of graphs without $ k $ node-disjoint odd cycles, where $ k $ is a constant.
Determining the complexity of the stable set problem for such graphs $ G $ is a long-standing problem.
We give a polynomial-time algorithm for the case that $ G $ can be further embedded in a (possibly non-orientable) surface of bounded genus.
Moreover, we obtain polynomial-size extended formulations for the respective stable set polytopes.

To this end, we show that $2$-sided odd cycles satisfy the Erd\H{o}s-P\'osa property in graphs embedded in a fixed surface.
This extends the fact that odd cycles satisfy the Erd\H{o}s-P\'osa property in graphs embedded in a fixed orientable surface (Kawarabayashi \& Nakamoto, 2007).

Eventually, our findings allow us to reduce the original problem to the problem of finding a minimum-cost non-negative integer circulation of a certain homology class, which turns out to be efficiently solvable in our case.
\end{abstract}

\maketitle

% ---
% ---
% ---
% ---
% ---

\section{Introduction}

The \emph{odd cycle packing number} of a graph $G$, denoted by $\ocp(G)$, is defined as the maximum number of node-disjoint odd cycles in $G$. Determining the complexity of the (maximum weight) stable set problem in graphs $G$ with bounded odd cycle packing number is a long-standing problem. It is known that the problem admits a PTAS whenever $ \ocp(G) $ is bounded by a constant. In fact, the stable set problem in graphs $G$ without an odd clique minor admits a PTAS, as shown by Tazari~\cite{tazari12}. Moreover, Bock, Faenza, Moldenhauer, and Ruiz-Vargas~\cite{BFMR14} state that the problem even admits a PTAS if $ \ocp(G) = O(\sqrt{|V(G)| / \log \log |V(G)|}) $. However, exact polynomial-time algorithms are only known for the cases $ \ocp(G) = 0 $ (bipartite graphs) and $ \ocp(G) = 1 $ (recently due to Artmann, Weismantel and Zenklusen~\cite{AWZ17} as explained below).

Let $ M \in \{0,1\}^{E(G) \times V(G)}$ denote the edge-node incidence matrix of $G$ and recall that the stable set problem can be written as determining $\max \, \{w^\intercal x \mid M x \leqslant \onevec,\ x \in \{0,1\}^{V(G)}\}$. Denoting by $\maxdet(M)$ the largest absolute value that the determinant of a square sub-matrix of $M$ can take, it is well-known that $\maxdet(M) = 2^{\ocp(G)}$ (see \cite{GKS95}). In particular, incidence matrices $ M $ of graphs with $ \ocp(G) \le 1$ satisfy $\maxdet(M) \le 2$. A main result in the recent work of Artmann \emph{et al.}~\cite{AWZ17} is that integer programs of this type can be solved in strongly polynomial time.

In light of recent work linking the complexity and structural properties of integer programs to the magnitude of their sub-determinants~\cite{Tardos86,DF94,VC09,BDEHN14,EV17,AWZ17,PSW}, it is tempting to believe that integer programs with bounded sub-determinants can be solved in polynomial time. This would imply in particular that the stable set problem on graphs with $\ocp(G) \leqslant k$ is polynomial for every fixed $k$. This remains open for $k \geqslant 2$.

Unfortunately, the approach of~\cite{AWZ17} provides limited insight on which properties of graphs with $ \ocp(G) \le 1 $ are relevant to derive efficient algorithms for graphs with higher odd cycle packing number. In contrast, Lov\'asz (see Seymour~\cite{seymour95}) gave a characterization of graphs without two node-disjoint odd cycles. The first published proof of Lov\'asz's theorem is due to Slilaty~\cite{slilaty07}. His proof uses matroid theory and extends Lov\'asz's theorem to signed graphs. Kawarabayashi \& Ozeki~\cite{KO13} later gave a short, purely graph-theoretical proof. A main implication is the following.

\begin{thm}[Lov\'asz, cited in~\cite{seymour95}] \label{thm:Lovasz}
Let $G$ be an internally $4$-connected graph without two node-disjoint odd cycles. Then
\begin{enumerate}[(i)]
    \item there is a set $X$ of at most three nodes such that $G - X$ is bipartite, or
    \item $G$ can be embedded in the projective plane such that every face is bounded by an even cycle.
\end{enumerate}
\end{thm}

While in case~(i) the stable set problem can be easily solved in polynomial time, case~(ii) is far more challenging.
In an embedding of the latter type, every odd cycle corresponds to a simple closed curve with a neighborhood homeomorphic to a M\"obius strip.  We call such curves \emph{$1$-sided}.
A cycle is \emph{$ 2 $-sided} if it is not $ 1 $-sided.
Since the projective plane does not contain two disjoint M\"obius strips, it does not contain two disjoint $1$-sided curves. It follows that graphs with an even face embedding in the projective plane do not have two node-disjoint odd cycles.

With this observation, one can construct non-trivial graphs with $\ocp(G) \leqslant k$ by considering non-orientable surfaces of Euler genus $ k $. If a graph $G$ is embedded in a non-orientable surface $\surf$ with Euler genus $k$ in such a way that every odd cycle of $G$ is a $1$-sided simple closed curve in $\surf$, then $\ocp(G) \leqslant k$. Unfortunately, no structural result similar to Theorem~\ref{thm:Lovasz} is known for higher odd cycle packing numbers. However, in view of Theorem~\ref{thm:Lovasz}, embeddings as described above yield the most natural non-trivial class of graphs with $ \ocp(G) \leqslant k $.

This motivates our study of graphs with bounded odd cycle packing number that can be embedded in a surface with bounded Euler genus. The following is our main result.

\begin{thm} \label{thm:main}
Let $ k,g $ be fixed. Then there is a polynomial-time algorithm for finding a maximum-weight stable set in any graph $G$ with $\ocp(G) \leqslant k$ embedded in a (possibly non-orientable) surface with Euler genus at most $g$.
\end{thm}

Let us mention that if we do not impose a constant bound on the odd cycle packing number, then the stable set problem is already NP-hard for $ g = 0 $ (planar graphs). In contrast, we do not know whether the above statement is still correct if we consider graphs with bounded odd cycle packing number but arbitrary Euler genus.

Our work also has implications on linear descriptions of stable set polytopes. Recall that an \emph{extended formulation} of a polytope $ P $ is a description of the form $ P = \{ x \mid \exists y : Ax + By \le b \} $ whose \emph{size} is the number of inequalities in $ Ax + By \le b $.

\begin{thm} \label{thm:extfor}
Let $ k,g $ be fixed. Then there is a polynomial-size extended formulation for the stable set polytope of any graph $G$ with $\ocp(G) \leqslant k$ embedded in a surface with Euler genus at most $g$.
\end{thm}

While proving Theorem~\ref{thm:main} and Theorem~\ref{thm:extfor}, we establish various intermediate results that may be of independent interest. Among them, we prove the following result.

\begin{thm} \label{EPodd2cycles}
There exists a computable function $f(g,k)$ such that for all graphs $G$ embedded in a surface with Euler genus $g$ and with no $k+1$ node-disjoint $2$-sided odd cycles, there exists $X \subseteq V(G)$ with $|X| \leq f(g,k)$ such that $G-X$ does not contain a $2$-sided odd cycle. 
Furthermore, there is such a set $X$ of size at most $19^{g+1}\cdot k$ if the surface is orientable.
\end{thm}

 Note that in a graph embedded in an orientable surface, every cycle is $2$-sided. Thus, the above statement extends the fact that odd cycles satisfy the Erd\H{o}s-P\'osa property in graphs embedded in a fixed orientable surface, which was proved by Kawarabayashi \& Nakamoto~\cite{KN_DM}. 
 We also remark that, for a fixed orientable surface, the upper bound in Theorem~\ref{EPodd2cycles} is linear in $k$, while the bound for orientable surfaces given in~\cite{KN_DM} is exponential in $k$. 
 A linear bound is best possible and answers a question from~\cite{KN_DM}.

 On the other hand, odd cycles in a graph embedded in a non-orientable surface do not satisfy the Erd\H{o}s-P\'osa property, as there are projective-planar graphs known as \emph{Escher walls} that do not contain two node-disjoint odd cycles but yet cannot be made bipartite by removing a constant number of nodes. In fact, Escher walls fall in case~(ii) of Theorem~\ref{thm:Lovasz}. These graphs were shown by Reed~\cite{Reed99} to be the sole obstruction for odd cycles to have the Erd\H{o}s-P\'osa property in general graphs.\footnote{We remark that the result of Kawarabayashi \& Nakamoto~\cite{KN_DM} also follows from Reed~\cite{Reed99} and the fact that the orientable genus of an Escher wall grows with its height, which follows from \cite{Mohar97}.}  

Another aspect of our work requires the study of the unbounded polyhedron
\[
    \sub(G) := \conv \{ x \in \Z^{V(G)} \mid Mx \le \onevec \},
\]
where $ M $ is again the edge-node incidence matrix of $ G $. We investigate several properties of this polyhedron and its relationship to the stable set polytope $ \stab(G) := \conv \{ x \in \{0,1\}^{V(G)} \mid Mx \le \onevec \} $ of $ G $. For instance, we show that every vertex of $ \sub(G) $ is still a $ 0/1 $-point. Furthermore, we establish the rather surprising identity $ \stab(G) = \sub(G) \cap [0,1]^{V(G)} $, which is false for general $ 0/1 $-matrices $ M $. We elaborate on these findings in Section~\ref{secDropConstraints} and Section~\ref{secEF}.

\subsection*{Organization of the paper}

We start our paper by recalling important notions and facts about surfaces and graphs embedded in surfaces in Section~\ref{secBasics}. In Section~\ref{sec:algorithm_overview} we describe our algorithm on a high level and state the relevant intermediate results. All steps of the algorithm are then explained in Sections~\ref{secNice}--\ref{secAlgorithm}. Proofs of all intermediate results are established in this part, except for Theorem~\ref{EPodd2cycles}.  The proof of Theorem~\ref{EPodd2cycles} is presented in the Section~\ref{secEP}, which completes the proof of Theorem~\ref{thm:main}. We discuss how to speed up parts of our algorithm in Section~\ref{sec:FPT}. In the final Section~\ref{secEF}, we demonstrate that our approach yields a polynomial-size extended formulation for stable set polytopes of graphs with bounded odd cycle packing number and bounded Euler genus, and hence establish the proof of Theorem~\ref{thm:extfor}.

\section{Basics on surfaces and graph embeddings}
\label{secBasics}

In this section, we recall important notions about surfaces and graphs embedded in surfaces. We refer the reader to the monograph by Mohar and Thomassen~\cite{MoharThom} for background material.

Throughout this paper, a \emph{surface} $ \surf $ is a non-empty compact connected Hausdorff topological space in which every point has a neighborhood that is homeomorphic to the plane $ \R^2 $. Let $\surf (h,c)$ be the surface obtained by cutting $2h+c$ disjoint open disks from the sphere and gluing $h$ cylinders and $c$ M\"obius strips onto the boundaries of these disks.  By the classification theorem for surfaces, every surface is homeomorphic to $\surf (h,c)$, for some $h$ and $c$.  The \emph{Euler genus} of a surface $\surf \cong \surf(h,c)$ is $2h+c$.

A simple closed curve in $\surf$ is said to be \emph{$1$-sided} if it has a neighborhood that is homeomorphic to a M\"obius strip, and \emph{$2$-sided} if it has a neighborhood that is homeomorphic to a cylinder. Every simple closed curve in $\surf$ is either $1$-sided or $2$-sided.
A surface $\surf$ is \emph{orientable} if all the simple closed curves in $\surf$ are $2$-sided.  Again by the classification theorem, a surface is uniquely determined by its Euler genus and whether or not it is orientable.

Consider a (finite, simple, undirected) graph $G$ embedded in $\surf$. The embedding maps the node set of $G$ to a set of distinct points of $\surf$ and the edge set of $G$ to a collection of internally disjoint simple open curves in $\surf$. For simplicity, we use $G$ both for the abstract graph and for the closed subset of $\surf$ that is the union of all points corresponding to nodes of $G$ and curves corresponding to edges of $G$.

The \emph{faces} of the embedded graph $G$ are the connected components of the open set $\surf \setminus G$. We denote the collection of faces of $G$ by $F(G)$. Each face $f$ of $G$ is a connected open set whose \emph{boundary} $\partial f$ is a subgraph of $G$.

A curve in $\surf$ is \emph{$G$-normal} if it only meets $G$ in nodes.  Its \emph{length} is the number of nodes of $G$ it intersects.  
A \emph{noose} is a $G$-normal, noncontractible, simple closed curved in $\surf$.
If $\surf$ is not the sphere, the \emph{facewidth} of the embedded graph $G$ is defined as the length of a shortest noose in $\surf$. Facewidth is not defined for graphs embedded in the sphere.

The embedding of $G$ in $\surf$ is \emph{cellular} if every face is homeomorphic to an open disk. Provided that $G$ is $2$-connected and $G$ is cellularly embedded in $\surf$ with facewidth at least $2$, the boundary $\partial f$ of each face $f$ is a cycle, see~\cite[Proposition 5.5.11]{MoharThom}.

Assume that $G$ is connected and cellularly embedded in $\surf$. Letting $g$ denote the Euler genus of $\surf$, a crucial fact is that
\begin{equation}
\label{eq:Euler_genus}
g = 2 + |E(G)| - |V(G)| - |F(G)|\,.
\end{equation}

A \emph{walk} in $G$ is a sequence $W = (v_0, e_1, v_1, \ldots, e_k, v_k)$ of nodes $v_i$ and edges $e_i$ such that for all $i \in [k]$, the edge $e_i$ has ends $v_{i-1}$ and $v_i$.
Let $W = (v_0, e_1, v_1, \ldots, e_k, v_k)$ be a walk.  The \emph{length} of $W$ is $k$, and $W$ is \emph{odd} or \emph{even} according to the parity of its length.  The \emph{ends} of $W$ are the nodes $v_0$ and $v_k$, which are respectively called the \emph{initial node} and the \emph{terminal node}. The walk $W$ is \emph{closed} if its ends coincide (that is, $v_0 = v_k$).  
The \emph{multiplicity} of node $v$ in $W$ is the number of indices $i$ such that $v_i = v$. In case $W$ is closed and $v = v_0 = v_k$, we subtract $1$ to define the multiplicity of $v$. A node is said to be \emph{repeated} if its multiplicity is larger than $1$. A walk is \emph{simple} if none of its nodes is repeated.

Let $W$ be a closed walk.  By regarding $W$ as an Eulerian multigraph, we can partition the edges of $W$ as $C_1 \cup \dots \cup C_\ell$, where each $C_i$ is a cycle.  We say that $W$ is \emph{$1$-sided} if the number of $1$-sided cycles among $C_1, \dots , C_\ell$ is odd; otherwise, $W$ is \emph{$2$-sided}.  We will see later that this definition does not depend on the choice of $C_1, \dots , C_\ell$.  For every face of $G$ one has a \emph{facial (closed) walk} that can be computed for instance using Edmonds' face traversal procedure~\cite{Edmonds60}. Facial walks of cellularly embedded graphs are always $2$-sided~\cite[Lemma 4.1.3]{MoharThom}.

\section{Algorithm overview}
\label{sec:algorithm_overview}

We start our paper by giving an overview of our algorithm, which consists of a sequence of polynomial-time reductions between different problems. Recall that we consider the following problem, in which we regard $ k $ and $ g $ as constants.

\begin{prob}
    \label{probOriginal}
    Given a graph $ G $ with $\ocp(G) \leqslant k$ embedded in a (possibly non-orientable) surface of Euler genus $ g $ and node weights $ w : V(G) \to \Q $, compute a maximum $ w $-weight stable set in $ G $.
\end{prob}

As a first step, we will apply some important preprocessing to obtain a more structured instance. The desired properties obtained after this step are summarized in the following definitions.

\begin{defn}
    Let $ G $ be a graph. We say that node weights $ w : V(G) \to \Q $ are \emph{edge-induced} (with respect to $ G $) if there exist non-negative edge weights $ w' : E(G) \to \Q_{\ge 0} $ such that $ w(v) = \sum_{e \in \delta(v)} w'(e) $ holds for all $ v \in V(G) $.
\end{defn}

\begin{defn}
    An embedding of a graph $G$ in a surface is \emph{parity-consistent} if it is cellular and if every odd closed walk is $ 1 $-sided.
\end{defn}

Notice that if a non-bipartite graph has a parity-consistent embedding in a surface $ \surf $, then $ \surf $ must be non-orientable. 

\begin{defn}
    A graph $G$ is said to satisfy the \emph{standard assumptions} if $G$ is $2$-connected, non-bipartite and embedded in a surface of Euler genus $g$ with a parity-consistent embedding. 
\end{defn}

Notice that if $G$ is parity-consistent embedded in a surface $\surf$ of Euler genus $g$, then every odd cycle of $G$ is $1$-sided.  Since $\surf$ contains at most $g$ disjoint $1$-sided curves, this implies $\ocp(G) \leqslant g$.  As we will show in Section~\ref{secNice}, we can efficiently reduce Problem~\ref{probOriginal} to the following, more restrictive problem.

\begin{prob}
    \label{probNice}
    Given a graph $G$ satisfying the standard assumptions and edge-induced node weights $ w : V(G) \to \Q $, compute a maximum $ w $-weight stable set in $ G $.
\end{prob}

We remark that one essential ingredient for reducing Problem~\ref{probOriginal} to Problem~\ref{probNice} will be Theorem~\ref{EPodd2cycles}, which we prove in Section~\ref{secEP}.

Recall that the stable set problem essentially asks for maximizing the linear function $ w(x) := \sum_{v \in V(G)} w(v) x(v) $ over $ \stab(G) $. In Section~\ref{secDropConstraints}, we will argue that Problem~\ref{probNice} is actually equivalent to the previous optimization task in which we replace $ \stab(G) $ by the unbounded polyhedron $ \sub(G) $.

\begin{prob}
    \label{probSubG}
    Given a graph $G$ satisfying the standard assumptions and edge-induced node weights $ w : V(G) \to \Q $, maximize $ w(x) $ over
    \[
        P(G) = \conv \{ x \in \Z^{V(G)} \mid x(v) + x(w) \le 1 \text{ for all } vw \in E(G) \}\,.
    \]
\end{prob}

By applying a suitable affine transformation from node space to edge space, we will show in Section~\ref{secSlackmap} that the above problem is equivalent to the following task. Here, for a walk $ W = (v_0,e_1,v_1,\dotsc,e_\ell,v_\ell) $ in $ G $ and a vector $ y \in \R^{E(G)} $, we use the notation
\[
    \omega_W(y) := \sum_{i=1}^\ell (-1)^{i-1} y(e_i)\,.
\]

\begin{prob}
    \label{probImageOfSubG}
    Given a graph $G$ satisfying the standard assumptions and non-negative edge costs $ c : E(G) \to \Q_{\ge 0} $, minimize $ c(y) := \sum_{e \in E(G)} c(e) y(e) $ over
    \[
        \slack(G) := \{ y \in \Z^{E(G)} : y \ge \zerovec, \, \omega_C(y) \text{ is odd}, \, \omega_W(y) = 0 \text{ for all even closed walks } W \}\,,
    \]
    where $ C $ is any fixed odd cycle in $ G $.
\end{prob}

We will see later that $ \slack(G) $ does not depend on the choice of $ C $.
Admittedly, the above description of the polyhedron $ \slack(G) $ is rather technical. As next steps, we will derive alternative descriptions of $ \slack(G) $ that appear to be more useful. To this end, we will equip the dual graph of $ G $ with a certain orientation. Note that the facial walk around a face $ f $ induces an order of the corresponding dual edges incident to the dual node $ f $.

\begin{defn}
    \label{defnDualOrientation}
    Let $G$ be a graph satisfying the standard assumptions. An orientation of the edges of the dual graph of $ G $ is called \emph{alternating} if in the local cyclic order of the edges incident to each dual node $ f $, the edges alternatively leave and enter $ f $.
\end{defn}

We will see that such an orientation always exists (and therefore it can be easily computed), provided that $ G $ satisfies the standard assumptions. Recall that there is a one-to-one correspondence between the edges of $ G $ and its dual graph. Thus, if $ D $ is an alternating orientation of the dual graph of $ G $, then we can view any vector $ y \in \R^{E(G)} $ also as a vector in the arc space $ \R^{A(D)} $ of $ D $. With this identification, we will see that the polyhedron $ \slack(G) $ coincides with the set of non-negative integer circulations in $ D $ that satisfy a few extra constraints. In fact, we will see that we can efficiently compute even closed walks $ W_1,\dotsc,W_{g-1} $ in $ G $ such that
\begin{align}
    \label{eqAlternativeDescriptionQ}
    \slack(G) = \conv \{ y \in \Z^{A(D)}_{\ge 0} \mid \ & y \text{ is a circulation in } D, \\
    \nonumber
    & \omega_C(y) \text{ is odd}, \, \omega_{W_1}(y) = \dotsb = \omega_{W_{g - 1}}(y) = 0 \}\,.
\end{align}

These claims will be established in Section~\ref{secDualOrientation}. For the sake of notation, let us bundle the above objects $ D, C, W_1,\dotsc,W_{g-1} $ using the following definition.

\begin{defn}
    \label{defnRepresentationSystem}
    Let $ G $ be a graph satisfying the standard assumptions. A \emph{dual representation} of $ G $ is an alternating orientation $ D $ of the dual graph of $ G $ together with the map $ \omega : \Z^{A(D)} \to \Z_2 \times \Z^{g - 1} $ defined via
    \[
        \omega(y) := (\omega_C(y) \ (\mathrm{mod} \, 2), \omega_{W_1}(y), \dotsc, \omega_{W_{g-1}}(y)) \in \Z_2 \times \Z^{g-1}\,,
    \]
    where $ C $ is an odd cycle in $ G $ and $ W_1,\dotsc,W_{g-1} $ are even closed walks in $ G $ that satisfy~\eqref{eqAlternativeDescriptionQ}.
\end{defn}

Note that if $ (D, \omega) $ is a dual representation of $ G $, then~\eqref{eqAlternativeDescriptionQ} reads
\[
    \slack(G) = \conv \{ y \in \Z^{A(D)}_{\ge 0} \mid y \text{ is a circulation in } D, \, \omega(y) = (1,\zerovec) \}\,.
\]

A final outcome of Section~\ref{secDualOrientation} will be the following:

\begin{prop}
    \label{propRepresentationSystem}
    Given a graph $G$ satisfying the standard assumptions, we can compute a dual representation of $G$ in polynomial time, provided that $ g $ is fixed.
\end{prop}

Although the constraint $ \omega(y) = (1,\zerovec) $ in the description of $ \slack(G) $ still looks technical, it turns out that the map $ \omega $ allows us to distinguish between \emph{homology classes} of integer circulations in $ D $. We will give a formal definition of this notion in Section~\ref{secCirc}. We will see that a non-negative integer circulation $ y $ of $ D $ is in $ \slack(G) $ if and only if it is homologous to the all-one circulation $ \onevec \in \Z^{A(D)} $, that is,
\[
    \slack(G) = \conv \{ y \in \Z^{A(D)}_{\ge 0} \mid y \text{ is a circulation in } D, \, y \text{ is homologous to } \onevec \}\,.
\]
These claims will be established in Section~\ref{secCirc}. As a consequence, Problem~\ref{probImageOfSubG} turns out to be equivalent to:
\begin{prob}
    \label{probCirculationHomologous}
    Given a graph $G$ satisfying the standard assumptions and non-negative edge costs $ c : E(G) \to \Q_{\ge 0} $, find a minimum $ c $-cost non-negative integer circulation in $ D $ that is homologous to $ \onevec \in \Z^{A(D)} $, where $ D $ is a dual orientation of $ G $.
\end{prob}
Finally, in Section~\ref{secAlgorithm} we will show that the above problem can be efficiently solved (provided that $ g $ is a constant) by reducing it to a sequence of min-cost flow problems.

\section{Standard assumptions and edge-induced weights}
\label{secNice}

In this section, we describe a procedure to reduce Problem~\ref{probOriginal} to Problem~\ref{probNice} that can be implemented to run in polynomial time, provided that $ g $ and $ k $ are constants. In the procedure, we will maintain a list of properties that we record as assumptions.

Let $ G $ be a graph with $\ocp(G) \leqslant k$ embedded in a surface of Euler genus $ g $ and let $ w : V(G) \to \Q $ be any node weights.

\begin{ass}
    \label{assOddCycles1Sided}
    Every odd closed walk in $ G $ is 1-sided.
\end{ass}

By Theorem~\ref{EPodd2cycles}, there is a set $ X $ of $ f(g) \cdot k$ nodes such that every odd cycle in 
$ G - X $ is 1-sided. By Lemma~\ref{EPodd2walks}, there is a set $Y$ of at most $ g $ nodes such that $G - X - Y$ has no $2$-sided odd closed walk. We can perform simple enumeration to find $X \cup Y$ and to reduce the original problem to a constant number of stable set problems in subgraphs of $ G - X -Y $. A more sophisticated algorithm is presented in Section~\ref{sec:FPT}. $ \diamond $

\begin{ass}
    \label{ass2Connected}
    $ G $ is $ 2 $-connected.
\end{ass}

If $ G $ is not $ 2 $-connected, it is an easy exercise to reduce the original problem to the computation of maximum-weight stable sets in at most $ 2|V(G)| $ subgraphs of $ G $ that are 2-connected. Clearly, any subgraph still satisfies Assumption~\ref{assOddCycles1Sided}. $ \diamond $

\begin{ass}
    \label{assEdgeInduced}
    $ w $ is edge-induced.
\end{ass}

Solve the linear program $ \max \{ w(x) \mid Mx \le \onevec, \, x \in [0,1]^{V(G)} \} $ to obtain an optimal solution $ x^* $. Let $ V_i := \{ v \in V(G) \mid x^*(v) = i \} $ for $ i \in \{0,1\} $. Nemhauser \& Trotter~\cite{NT74} showed that there exists a maximum weight stable set $ S^* $ with $ S^* \cap (V_0 \cup V_1) = V_1 $. Thus, we can delete the nodes in $ V_0 \cup V_1 $ and may assume that $ 0 < x^*(v) < 1 $ holds for all $ v \in V(G) $.

This implies that $ x^* $ is also optimal for $ \max \{ w(x) \mid Mx \le \onevec \} $. By LP-duality, this means that $ w $ (treated as a vector in $ \R^{V(G)} $) is a conic combination of the rows of $ M $, which yields the claim.

Note that the deletion of nodes might create a graph that violates Assumption~\ref{ass2Connected} (but that still satisfies Assumption~\ref{assOddCycles1Sided}). In this case, we repeat with the previous step. Since nodes get deleted whenever we repeat this process, we have to spend at most $ |V(G)| $ iterations on this. $ \diamond $

\begin{ass}
    \label{assNonBipartite}
    $ G $ is non-bipartite.
\end{ass}

It is well-known that if $ G $ is bipartite, then the weighted stable set problem can be solved in polynomial time, which means that we are done in this case. $ \diamond $

\begin{consq}
    \label{consqNonOrientable}
    $ \surf $ is non-orientable.
\end{consq}

If $ \surf $ is orientable, then every cycle in the embedding of $ G $ is 2-sided. By Assumption~\ref{assOddCycles1Sided} this means that $ G $ contains no odd cycle and hence is bipartite, a contradiction to Assumption~\ref{assNonBipartite}. $ \diamond $

\begin{ass}
    \label{assMinGenus}
    $ G $ cannot be embedded in a surface of Euler genus $ g - 1 $.
\end{ass}

For each fixed surface, one can efficiently compute whether a graph embeds into it and if yes, compute such an embedding, see~\cite{KMR}. Thus, we can check whether $ G $ can be embedded into a surface of Euler genus $ g - 1 $. If yes, we embed $ G $ into that surface and repeat all steps described above. Note that the total process will be repeated at most $ g $ times. $ \diamond $

\begin{consq} \label{consqCellular}
    The embedding of $ G $ is cellular.
\end{consq}

Assumption~\ref{assMinGenus} means that $ G $ is minimum-genus embedded. It is well-known that such embeddings satisfy the above property, see~\cite{PPTV}. (See also \cite[Propositions 3.4.1 and 3.4.2]{MoharThom}.) $ \diamond $

\section{Dropping constraints}
\label{secDropConstraints}

Next, we argue that Problem~\ref{probNice} is equivalent to Problem~\ref{probSubG}. To this end, we will make use of the following result.

\begin{prop} \label{prop:subG_01-vertices}
Every vertex of $ \sub(G) $ is a $ 0/1 $-point.
\end{prop}
\begin{proof}
Let $x^*$ be a vertex of $\sub(G)$. Observe that $ x^* $ is integer. Let $ E(x^*) := \{ vw \in E \mid x^*(v) + x^*(w) = 1 \} $ be the set of ``tight'' edges, and let $ G' = (V', E') $ be a connected component of the graph $ (V,E(x^*)) $. Then $ G' = (V', E') $ is bipartite, and there exists an integer $ \alpha \ge 1 $ such that  $V'_{\alpha} := \{ v \in V' \mid  x^*(v) = \alpha \} $ and $ V'_{1-\alpha} := \{ v \in V' \mid x^*(v) = 1 - \alpha \} $ is a bipartition of $G'$.  

Assume now that $ x^* \notin \{0,1\}^V $. Then we can assume that $ \alpha \ge 2 $. Let $ r \in \Z^V $ be the vector such that $ r(v) := 1 $ for all $ v \in V'_{\alpha} $, $r(v) := -1 $ for all $ v \in V'_{1 - \alpha} $, and $ r(v) := 0 $ for all $ v \in V \setminus V' $. We see that both $ x^* + r $ and $ x^* - r $ are integer points in $ \sub(G) $, contradicting that $ x^*$  is a vertex.
\end{proof}

In order to see that Problem~\ref{probNice} and Problem~\ref{probSubG} are equivalent, it suffices to show the following.
\begin{prop}
    \label{propDropConstraints}
    Let $ G $ be a graph with no bipartite connected component and let $ w : V(G) \to \R $ be edge-induced node weights. Then
    \[
        \max \left\{ w(x) \mid Mx \le \onevec, \, x \in \{0,1\}^{V(G)} \right\} = \max \left\{ w(x) \mid Mx \le \onevec, \, x \in \Z^{V(G)} \right\}\,,
    \]
    where $ M \in \{0,1\}^{E(G) \times V(G)} $ is an edge-node incidence matrix of $ G $.
\end{prop}
\begin{proof}
    It suffices to show that $ \max \{ w(x) \mid x \in \sub(G) \} $ exists and that it is attained in a $ 0/1 $-point of $ \sub(G) $.

    Since $ w $ is edge-induced, there exist edge costs $ c : E(G) \to \R_{\ge 0} $ such that $ w(v) = \sum_{e \in \delta(v)} c(e) $ holds for all $ v \in V(G) $. For any vector $ x \in \R^{V(G)} $ that satisfies $ Mx \le \onevec $ we have
    \[
        w(x) = \sum_{v \in V(G)} w(v) x(v) = \sum_{vw \in E(G)} c(vw)(\underbrace{x(v) + x(w)}_{\le 1}) \le \sum_{e \in E(G)} c(e)\,.
    \]
    This means that $ \sup \{ w(x) \mid x \in \sub(G) \} $ is finite and hence the maximum is attained.

    Since no connected component of $ G $ is bipartite, $ M $ has full column-rank. This implies that $ \{ x \in \R^{V(G)} \mid Mx \le \onevec \} $ is pointed, and so is $ \sub(G) $. Hence, $ \max \{ w(x) \mid x \in \sub(G) \} $ is attained in a vertex of $ \sub(G) $, which, by Proposition~\ref{prop:subG_01-vertices} is a $ 0/1 $-point.
\end{proof}

\section{From node space to edge space}
\label{secSlackmap}

In this section, we will show that Problem~\ref{probSubG} is equivalent to Problem~\ref{probImageOfSubG}. To this end, consider the affine map $ \sigma : \R^{V(G)} \to \R^{E(G)} $ defined via
\[
    \sigma(x) := \onevec - Mx,
\]
where $ M $ is again a node-edge incidence matrix of $ G $. In other words, a vector $ x \in \R^{V(G)} $ is mapped to $ y = \sigma(x) \in \R^{E(G)} $ where $ y(vw) = 1 - x(v) - x(w) $ for every edge $ vw \in E(G) $.

Suppose that $ G $ is connected and non-bipartite. In this case, $ M $ has full column-rank and hence $ \sigma $ is injective. In particular, $ \sigma $ defines a bijection between $ \sub(G) $ and $ \sigma(\sub(G)) $. In this case it is clear that there exists an affine objective function assigning each vector in $ \sigma(\R^{E(G)}) $ the weight of its preimage in $ \R^{V(G)} $ under $ w $. This means that Problem~\ref{probSubG} is equivalent to optimizing a linear function over $ \sigma(\sub(G)) $. While the new objective function may not be unique, there is a natural choice for it in case of edge-induced weights $ w $:

\begin{obs}
    If the node weights $ w : \R^{V(G)} \to \R $ are edge-induced by $ c : E(G) \to \R_{\ge 0} $, then we have $ w(x) = c(E(G)) - c(y) $ for all $ x \in \R^{V(G)} $ and $ y = \sigma(x) $.
\end{obs}

Thus, maximizing $ w $ over $ \sub(G) $ is equivalent to minimizing $ c $ over $ \sigma(\sub(G)) $. Note also that $ c $ is non-negative. It remains to show the following.

\begin{prop}
    \label{propImageSubG}
    Let $ G $ be a connected graph and let $ C $ be an odd cycle in $ G $. Then
    \begin{align*}
        \sigma(\sub(G)) & = \conv \{ y \in \Z^{E(G)}_{\ge 0} \mid \omega_C(y) \text{ is odd}, \, \omega_W(y) = 0 \text{ for all even closed walks } W \} \\
        & = \slack(G) \,.
    \end{align*}
\end{prop}

Recall that for a walk $ W = (v_0,e_1,v_1,\dotsc,e_\ell,v_\ell) $ in $ G $ and a vector $ y \in \R^{E(G)} $, we defined
\[
    \omega_W(y) = \sum_{i=1}^\ell (-1)^{i-1} y(e_i)\,.
\]

A \emph{$1$-tree} is a graph that is a tree plus one extra edge. A $1$-tree is said to be \emph{odd} if its unique cycle is odd. The following basic observation will be used several times below.

\begin{obs} \label{obs:walks}
    Let $G$ be a connected non-bipartite graph, let $T$ be a spanning odd $1$-tree of $G$. For every edge $e \in E(G) \setminus E(T)$ there is both an odd closed walk $W_e^{\mathrm{odd}}$ in $T + e$ that contains $e$ exactly once and an even closed walk $W_e^{\mathrm{even}}$ in $T + e$ that contains $e$ exactly once.
\end{obs}

In the proof of Proposition~\ref{propImageSubG} we will make use of the following lemma.

\begin{lem}
    \label{lemImageNodeSpace}
    For any connected non-bipartite graph $ G $ we have
    \[
        \sigma(\R^{V(G)}) = \{ y \in \R^{E(G)} \mid \omega_W(y) = 0 \text{ for all even closed walks } W \}\,.
    \]
\end{lem}
\begin{proof}
    Let $ Y $ denote the subspace on the right-hand side of the above equation. To see that $ \sigma(\R^{V(G)}) \subseteq Y $ let $ x \in \R^{V(G)} $ and set $ y := \sigma(x) $. For every even walk $W = (v_0,e_1,v_1, \ldots, e_{2k},v_{2k})$, we have
    \begin{align*}
        \omega_W(y)
        &= \sum_{i=1}^{2k} (-1)^{i-1} y(e_i)\\
        &= \sum_{i=1}^{2k} (-1)^{i-1} (1 - x(v_{i-1}) - x(v_{i}))\\
        &= x(v_{2k}) - x(v_0)\,.
    \end{align*}
    If $W$ is closed, then $v_{2k} = v_0$ and thus $\omega_W(y) = 0$. Hence, we obtain $ y \in Y $.

    Next, we show that $ \dim(Y) \le \dim(\sigma(\R^{V(G)})) $, which yields $ \sigma(\R^{V(G)}) = Y $. To this end, recall that $ \sigma $ is injective since $ G $ is connected and non-bipartite, and hence $ \dim(\sigma(\R^{V(G)})) = \dim(\R^{V(G)}) = |V(G)| $. Thus, it remains to show $ \dim(Y) \le |V(G)| $.

    Let $T$ be a spanning odd $1$-tree of $G$. By Observation~\ref{obs:walks}, for each edge $ e \in E(G) \setminus E(T) $ there is an even closed walk $W_e^{\mathrm{even}}$ in $T + e$ traversing $e$ exactly once.

    In this way we can construct a family of even closed walks $ \mathcal{W} $ such that for every edge $ e \in E(G) \setminus E(T) $ there is exactly one walk in $ \mathcal{W} $ that contains it. Thus, the corresponding equations $ \omega_W(y) = 0 $, $ W \in \mathcal{W} $ are linearly independent and hence
    \[
        \dim(Y) \le |E(G)| - |\mathcal{W}| = |E(G)| - (|E(G) \setminus E(T)|) = |E(T)| = |V(G)|\,. \qedhere
    \]
\end{proof}

We leave to the reader to check that Lemma~\ref{lemImageNodeSpace} is actually true for all graphs.

\begin{proof}[Proof of Proposition~\ref{propImageSubG}]
    It suffices to show that the two sets
    \begin{align*}
        A & := \{ x \in \Z^{V(G)} \mid Mx \le \onevec \}, \\
        B & := \{ y \in \Z^{E(G)} \mid y \ge \zerovec, \, \omega_C(y) \text{ is odd}, \, \omega_W(y) = 0 \text{ for all even closed walks } W \}
    \end{align*}
    satisfy $ \sigma(A) = B $.

    To see that $ \sigma(A) \subseteq B $, let $ x \in A $ and consider $ y := \sigma(x) = \onevec - Mx $. It is clear that $ y \ge \zerovec $ and by Lemma~\ref{lemImageNodeSpace} we also have $ \omega_W(y) = 0 $ for every even closed walk $ W $. Furthermore, we have
    \[
        \omega_C(y) \equiv \sum_{vw \in E(C)} y(vw) = \sum_{vw \in E(C)} (1 - x(v) - x(w)) = |C| - \sum_{v \in V(C)} 2 x(v) \equiv 1 \pmod 2\,.
    \]
    To see that $ B \subseteq \sigma(A) $, let $ y \in B $. By Lemma~\ref{lemImageNodeSpace} we know that there exists a vector $ x \in \R^{V(G)} $ with $ \sigma(x) = y $. Since $ y $ is non-negative, we clearly have $ Mx \le \onevec $. It remains to show that $ x $ is integer. To this end, let $e_1 = v_0v_1$, $e_2 = v_1v_2$, \ldots, $e_{2k+1} = v_{2k}v_{2k+1}$ denote the edges of $C$, where $v_{2k+1} = v_0$. We have
    \[
        \sum_{i=1}^{2k+1} (-1)^{i-1} y(e_i) = \sum_{i=1}^{2k+1} (-1)^{i-1} (1 - x(v_{i-1}) - x(v_{i})) = 1 - 2x(v_0).
    \]
    Since the left-hand side is odd, we conclude that $x(v_0)$ is integer. Since $ 1 - x(v) - x(w) = y(vw) $ is an integer for all edges $vw$ of $G$ and since $ G $ is connected, this implies that $x(v)$ is integer for all nodes $v$ of $G$.
\end{proof}

\section{Towards a circulation problem}
\label{secDualOrientation}

So far, we have seen that we can efficiently reduce Problem~\ref{probOriginal} to Problem~\ref{probImageOfSubG}, in which we are given a graph $G$ satisfying the standard assumptions and a linear objective with non-negative coefficients that we want to minimize over the polyhedron
\begin{align*}
    \slack(G) & = \sigma(\sub(G)) \\
    & = \conv \{ y \in \Z^{E(G)}_{\ge 0} \mid \omega_C(y) \text{ is odd}, \, \omega_W(y) = 0 \text{ for all even closed walks } W \}\,,
\end{align*}
where $ C $ is any odd cycle in $ G $. The aim of this section is to provide an alternative description of $ \slack(G) $ in the sense of Equation~\eqref{eqAlternativeDescriptionQ}. To this end, let us consider the dual graph of $ G $, which we will equip with a certain orientation:

\begin{lem}
    \label{lemDualAlternatingOrientation}
    Given a graph $G$ satisfying the standard assumptions, there exists an alternating orientation $ D $ of the dual graph of $ G $.
\end{lem}

Recall that an orientation of the edges of the dual graph of $ G $ is called \emph{alternating} if in the local cyclic ordering of the edges incident to each dual node $ f $, the edges alternatively leave and enter $ f $, see Definition~\ref{defnDualOrientation}. Note that the existence of the alternating orientation $ D $ also implies that we can efficiently compute it.

Before we provide a proof of Lemma~\ref{lemDualAlternatingOrientation}, let us demonstrate that this yields an alternative description of $ \slack(G) $. Recall that the edges of the dual graph are associated to the edges of $ G $. So, given an alternating orientation $ D $ of the dual graph, we can identify points in $ \Z^{E(G)} $ with points in $ \Z^{A(D)} $.

\begin{lem}
    \label{lemQofGWithCirculations}
    Let $G$ be a graph satisfying the standard assumptions, $ C $ be an odd cycle in $ G $, and $ D $ be an alternating orientation of the dual graph of $ G $. Then we can efficiently find even closed walks $ W_1,\dotsc,W_{g-1} $ in $ G $ such that
    \begin{align*}
        \slack(G) = \conv \{ y \in \Z^{A(D)}_{\ge 0} \mid \ & y \text{ is a circulation in } D, \\
        & \omega_C(y) \text{ is odd}, \, \omega_{W_1}(y) = \dotsb = \omega_{W_{g - 1}}(y) = 0 \}\,.
    \end{align*}
\end{lem}

Note that Lemmas~\ref{lemDualAlternatingOrientation} and~\ref{lemQofGWithCirculations} imply Proposition~\ref{propRepresentationSystem}.

\begin{proof}[Proof of Lemma~\ref{lemQofGWithCirculations}]
    Let $ Z $ be the subspace of $ \R^{A(D)} $ of all circulations in $ D $ and consider
    \[
        L := \conv \{ y \in \R^{E(G)} \mid \omega_W(y) = 0 \text{ for all even closed walks } W \}
    \]
    Note that it suffices to show that there exist even closed walks $ W_1,\dotsc,W_{g-1} $ in $ G $ (that can be efficiently computed) such that
    \begin{equation}
        \label{eqSanZeno}
        L = \{ y \in Z \mid \omega_{W_1}(y) = \dotsb = \omega_{W_{g - 1}}(y) = 0 \}\,,
    \end{equation}
    where we identify $ \R^{E(G)} $ with $ \R^{A(D)} $.

    We first claim that $ L \subseteq Z $. To see this, let $ y \in L $ and consider any node $ f $ of $ D $. Recall that $ f $ corresponds to a face of $ G $ and let $ W $ be a facial walk of $ f $. Since $ D $ is alternating and $ W $ is even, we have $ \delta^{\text{in}}(y) - \delta^{\text{out}}(y) = \pm \omega_W(y) = 0 $. Thus, $ y $ is indeed a circulation.

    By Lemma~\eqref{lemImageNodeSpace} we know that $ L $ is equal to $ \sigma(\R^{V(G)}) $ and hence
    \[
        \dim(L) = |V(G)|\,.
    \]
    On the other hand, it is a basic fact that
    \[
        \dim(Z) = |A(D)| - |V(D)| + 1 = |E(G)| - |F(G)| + 1 = |V(G)| + g - 1\,,
    \]
    where the last equality follows from the fact that $ G $ is cellularly embedded into a surface of Euler genus $ g $ and Equation~\eqref{eq:Euler_genus}.

    Thus, we can obtain $ L $ by adding $ g - 1 $ equations (from the definition of $ L $) to the definition of $ Z $. In other words, there must exist even closed walks $ W_1,\dotsc,W_{g-1} $ such that Equation~\eqref{eqSanZeno} holds. Using the construction given in the proof of Lemma~\ref{lemImageNodeSpace} it is also easy to compute even closed walks $ W_1,\dotsc,W_{g-1} $ as required.
\end{proof}

We are almost ready to prove Lemma~\ref{lemDualAlternatingOrientation}.  However, before doing so, we require a more combinatorial definition of $1$-sided and $2$-sided closed walks.  Let $G$ be a graph embedded in a surface $\surf$.  Regardless of the (global) orientability of $\surf$, one can define a \emph{local orientation} around each node $v$ of $G$. We point out that $G$ has $2^{|V(G)|}$ systems of local orientations. If $G$ has minimum degree at least $3$, this translates nicely into the combinatorial notion of a rotation system, see~\cite[Chapters 3 and 4]{MoharThom}. However, here we stick to the topological view to avoid the technicalities arising with nodes of degree $1$ and $2$.

Each choice of local orientations, gives rise to a \emph{signature} $\Sigma \subseteq E(G)$ obtained as follows. If the local orientations at the ends of $e$ are inconsistent (one clockwise and the other counterclockwise), then $e$ is included in $\Sigma$. Otherwise, $e$ is not included in $\Sigma$. See Figure~\ref{figLocalOrientations} for an illustration. For every signature $\Sigma$, a cycle $C$ in $G$ (seen as a simple closed curve in $\surf$) is $1$-sided if and only if $|E(C) \cap \Sigma| \equiv 1 \pmod{2}$. We stress that this holds no matter which local orientations are used to define the signature.  

\begin{figure}
    \begin{center}
        \includegraphics{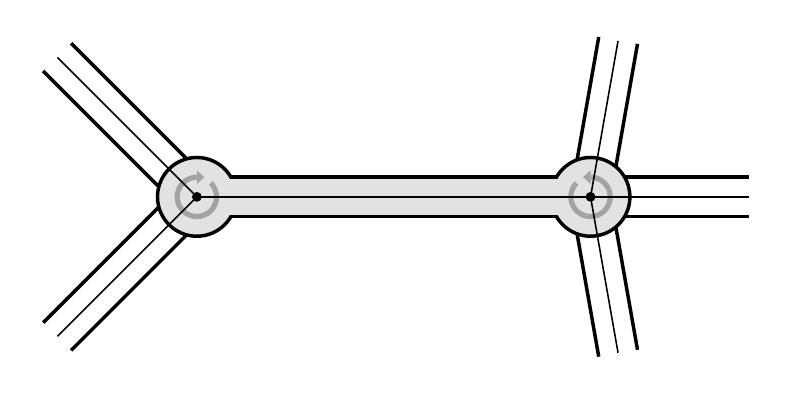}
    \end{center}
    \caption{Using the neighborhood of an edge, one can compare local orientations around its ends. In the above case, the local orientations are inconsistent and hence the edge belongs to $ \Sigma $.}
    \label{figLocalOrientations}
\end{figure}

Given $\Sigma \subseteq E(G)$, a walk $W = (v_0, e_1, v_1, \ldots, e_k, v_k)$ is said to be \emph{$\Sigma$-odd} if $|\{i \in [k] : e_i \in \Sigma\}|$ is odd, and \emph{$\Sigma$-even} otherwise. Let $G$ be embedded in a surface $\surf$, and let $\Sigma$ be a signature of the embedding. We say that a walk in $G$ is \emph{$1$-sided} if it is $\Sigma$-odd, and \emph{$2$-sided} if it is $\Sigma$-even.  When restricted to closed walks, this definition of $1$-sided and $2$-sided agrees with the definition given in Section~\ref{secBasics}.  

Before proving Lemma~\ref{lemDualAlternatingOrientation}, we need the following lemma. Here, we denote the symmetric difference of two sets $ A, B $ by $ A \triangle B $.

\begin{lem}
    \label{lemAuxSignatureIsWholeEdgeSet}
    Let $ G $ be a connected non-bipartite graph and $ \Sigma \subseteq E(G) $ be such that every odd closed walk is $\Sigma$-odd. Then there exist nodes $ v_1,\dotsc,v_k \in V(G) $ such that $ \Sigma \triangle \delta(v_1) \triangle \dotsb \triangle \delta(v_k) = E(G) $.
\end{lem}
\begin{proof}
    Let $T$ be a spanning odd $1$-tree of $G$. Let $e'$ denote any edge on the odd cycle of $T$. It is easy to find nodes $ v_1,\dotsc,v_k $ such that $ \Sigma' := \Sigma \triangle \delta(v_1) \triangle \dotsb \triangle \delta(v_k) $ contains every edge of the spanning tree $ T-e' $. Notice that a closed walk is $\Sigma'$-odd if and only if it is $\Sigma$-odd, as can be seen by induction on $k$.

    Now consider each edge $e \in E(G) \setminus E(T-e')$ one by one, starting with $e'$. For each choice of $e$ there is an odd closed walk $W_e^{\mathrm{odd}}$ containing $e$ exactly once for which all the edges except possibly $e$ are known to be in $\Sigma'$. For $e = e'$, we may take $W_e^{\mathrm{odd}}$ to be the unique odd cycle of $T$. For subsequent edges, we invoke Observation~\ref{obs:walks}. Since $W_e^{\mathrm{odd}}$ is odd, it should be $\Sigma$-odd, and thus $\Sigma'$-odd. Since all its edges except possibly one are in $\Sigma'$, this forces $e$ to belong to $\Sigma'$.
\end{proof}

\begin{proof}[Proof of Lemma~\ref{lemDualAlternatingOrientation}]
    For each node of $ G $, we choose a local orientation around it. This defines a signature $ \Sigma \subseteq E(G) $. Since the embedding of $ G $ is parity-consistent, we have that every odd closed walk is $\Sigma$-odd. By Lemma~\ref{lemAuxSignatureIsWholeEdgeSet} there exist nodes $ v_1,\dotsc,v_k \in V(G) $ such that $ \Sigma \triangle \delta(v_1) \triangle \dotsb \triangle \delta(v_k) = E(G) $. This means that if we reverse the local orientations around the nodes $ v_1,\dotsc,v_k $, the induced signature changes to $ \Sigma = E(G) $.

    In this case, the local orientations around the ends of each edge must be inconsistent. This implies a natural orientation $ D $ of the edges of the dual graph as depicted in Figure~\ref{figDualOrientation}. By construction, if we now traverse the facial walk of some face $ f $, then the arcs corresponding to the edges of the walk alternatively enter and leave $ f $. This shows that $ D $ is an alternating orientation in the sense of Definition~\ref{defnDualOrientation}. \qedhere

    \begin{figure}
        \begin{center}
            \includegraphics{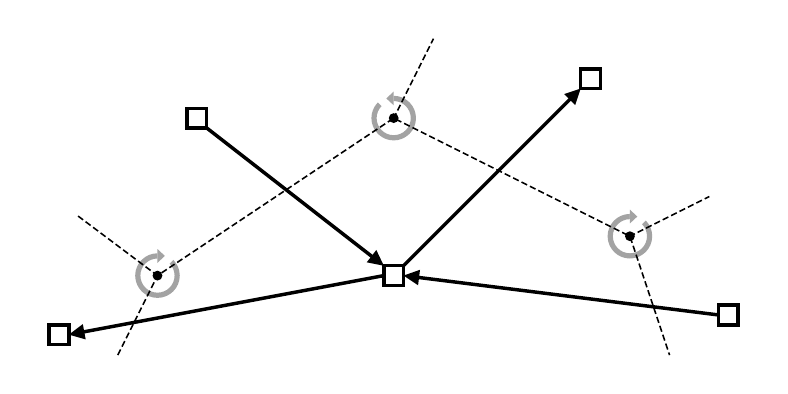}
        \end{center}
        \caption{If the local orientations around the nodes of the graph (dots and dashed edges) are inconsistent along edges, then they induce an orientation of the dual graph (squares and solid arcs).}
        \label{figDualOrientation}
    \end{figure}
\end{proof}

\section{Homologous circulations}
\label{secCirc}

Given a graph $ G $ satisfying the standard assumptions, we have seen that we can efficiently compute a dual representation $ (D, \omega) $ of $ G $ that allows us to describe $ \slack(G) $ via
\[
    \slack(G) = \conv \{ y \in \Z^{A(D)}_{\ge 0} \mid y \text{ is a circulation of } D, \, \omega(y) = (1,\zerovec) \}\,.
\]
Recall that $ D $ is an alternating orientation of the dual of $ G $, and $ \omega $ arises from an odd cycle $ C $ of $ G $, and even closed walks $ W_1,\dotsc,W_{g-1} $ in $ G $ such that
\[
    \omega(y) = (\omega_C(y) \ (\mathrm{mod} \, 2), \omega_{W_1}(y), \dotsc, \omega_{W_{g-1}}(y)) \in \Z_2 \times \Z^{g-1}\,,
\]
see Definition~\ref{defnRepresentationSystem}. In what follows, we will show that the constraint $ \omega(y) = (1,\zerovec) $ has a clean topological interpretation.

To this end, we make use of the \emph{first homology group of $\surf$ over the integers}: The cellular embedding of $ G $ induces a cellular embedding of $ D $. Consider a facial walk of $ D $. Traversing this walk in one (arbitrary) direction, we obtain a circulation in $ \{0,1\}^{A(D)} $, which we call a \emph{facial circulation}. We declare two integer circulations $ y_1, y_2 \in \Z^{A(D)} $ \emph{homologous} if $ y_1 - y_2 $ is a linear combination of facial circulations with integer coefficients. It turns out that $ \omega $ allows us to distinguish between the equivalence classes of integer circulations in $ D $. We remark that the fact that $ \omega $ is linear in each component will be (implicitly) used several times in what follows.

\begin{prop}
    \label{propOmegaDistinguishHomologyClass}
    Let $G$ be a graph satisfying the standard assumptions and let $ (D,\omega) $ be a dual representation of $ G $. For any integer circulations $ y_1, y_2 $ in $ D $ we have $ \omega(y_1) = \omega(y_2) $ if and only if $ y_1 $ is homologous to $ y_2 $.
\end{prop}

\begin{proof}
    Let $ y_1, y_2 $ be integer circulations and set $ y := y_1 - y_2 $. Suppose first that $ y_1 $ is homologous to $ y_2 $, meaning that $ y $ is a linear combination of facial circulations with integer coefficients. Notice that facial circulations in $D$ correspond to stars $\delta(v) = \delta_G(v)$ in $G$. Now $C$ contains an even number of edges (zero or two) in each star $\delta(v)$. This shows that $\omega_C$ is zero modulo $2$ on all facial circulations and thus $\omega_C(y) = 0 $.

    Similarly, for $i \in [g-1]$, each $W_i$ enters and leaves a node $ v \in V(G) $ the same number of times, and hence for each star $ \delta(v) $, half of its edges appear with a positive sign and half of its edges with a negative sign in $ \omega_{W_i} $, respectively. Thus $\omega_{W_i}$ is identically zero on all facial circulations. This shows $\omega_{W_i}(y) = 0 $ for all $ i \in [g-1] $ and hence $ \omega(y) = (0,\zerovec) $. This means that $ \omega(y_1) = \omega(y_2) $.

    Now suppose that $ \omega(y_1) = \omega(y_2) $. This means that $ \omega(y) = (0,\zerovec) $. In other words, $\omega_C(y)$ is even and $\omega_{W_i}(y)$ is zero for all $i \in [g-1]$. From the proof of Lemma~\ref{lemImageNodeSpace} we see that there exists a point $ x \in \Z^{V(G)} + \frac{1}{2}\mathbf{1}$ such that $\sigma(x) = y$. We can write $x$ as an affine combination of the points $\frac{1}{2}\mathbf{1}$ and $\frac{1}{2} \mathbf{1} - \mathbf{e}_v$ for $v \in V(G)$, with integer coefficients. Applying $\sigma$ to this integer affine combination, we see that $y$ is an integer combination of facial circulations. This means that $ y_1 $ is homologous to $ y_2 $.
\end{proof}

Recall that a non-negative integer circulation $ y $ in $ D $ is contained in $ \slack(G) $ if and only if $ \omega(y) = (1,\zerovec) $. Furthermore, note that the all-one circulation $ \onevec = \sigma(\zerovec) $ is contained in $ \slack(G) $. Thus, we obtain:

\begin{cor}
    Let $G$ be a graph satisfying the standard assumptions and let $ (D,\omega) $ be a dual representation of $ G $. Then we have
    \[
        \slack(G) = \conv \{ y \in \Z^{A(D)}_{\ge 0} \mid y \text{ is a circulation in } D, \, y \text{ is homologous to } \onevec \}\,.
    \]
\end{cor}

This shows that Problem~\ref{probImageOfSubG} is indeed equivalent to Problem~\ref{probCirculationHomologous}. Let us remark that since $ D $ is cellularly embedded in $\surf$, the first homology group of $\surf$ over the integers, $H_1(\surf; \Z)$, is the quotient of the group of all integer circulations $y$ in $ D $ by the subgroup of circulations that are homologous to the zero circulation\footnote{We point out that, in order to define the homology group $H_1(\surf; \Z)$ in this way, $D$ can be any orientation of an undirected graph that is cellularly embedded in $\surf$. The resulting group depends only on $\surf$, and not on the specific directed graph $D$.}. It is well-known that, since $\surf$ is non-orientable with Euler genus $g$, $H_1(\mathbb{S};\Z)$ is isomorphic to $\Z_2 \times \Z^{g-1}$~\cite[Chapter 8.3]{Armstrong}. It is not hard to see that $ \omega(y) = (1,\zerovec) $ actually means that the homology class of $ y $ is equal to $ (1,\zerovec) $ in $H_1(\mathbb{S};\Z)$. However, we will not need this in our algorithm.

\section{Solving Problem~\ref{probCirculationHomologous} in polynomial time}
\label{secAlgorithm}

In this section we provide a polynomial-time algorithm to solve Problem~\ref{probCirculationHomologous}, for constant Euler genus $g$. We mention that a similar problem was already solved by Chambers, Erickson and Nayyeri~\cite{CEN12}. However, since it is not clear to us that their algorithm can be adapted to our setting (see Section~\ref{sec:FPT} for further details), we describe a self-contained algorithm here.  Our approach consists of a sequence of circulation problems in auxiliary graphs. 

As a first step, we require the following topological lemma. 

\begin{lem} \label{lem:disjointhomology}
    Let $D$ be a directed graph cellularly embedded in a surface $\surf$. If $C_1, \dotsc, C_\ell$ are pairwise node-disjoint directed cycles of $D$ whose corresponding whose homology classes are pairwise distinct, then $ \ell \leqslant 6g(\surf)$.
\end{lem}

To prove Lemma~\ref{lem:disjointhomology}, we use a result of Malni\v{c} and Mohar~\cite[Proposition 3.1]{MM92}, which requires a few more definitions.  Let $\gamma_0$ and $\gamma_1$ be closed curves in a surface $\surf$.  We say that $\gamma_0$ and $\gamma_1$ are \emph{freely homotopic} if there is a continuous function $h:[0,1] \times [0,1] \to \surf$ such that for all $t \in [0,1]$, $h(0,t)=\gamma_0(t)$ and $h(1,t)=\gamma_1(t)$; and for each $s \in [0,1]$, $h(s,0)=h(s,1)$.  The closed curve $\gamma_0^{-1}$ is the curve given by $\gamma_0^{-1}(t) := \gamma_0(1-t)$ for all $t \in [0,1]$. 
\begin{lem}[{\cite[Proposition 3.1]{MM92}}]
    \label{lem:MM92}
     Let $\gamma_1,\dotsc,\gamma_\ell$ be disjoint noncontractible simple closed curves in a surface $\surf$.  If for all distinct $i,j \in [\ell]$, $\gamma_i$ is not freely homotopic to $\gamma_j$ nor $\gamma_j^{-1}$, then $ \ell \leqslant \max \{1, 3(g(\surf)-1)\}$.
\end{lem}

Lemma~\ref{lem:disjointhomology} now follows easily.  

\begin{proof}[Proof of Lemma~\ref{lem:disjointhomology}] For $i \in [\ell]$ we denote by $\gamma_i : [0,1] \to \surf$ the simple closed curve corresponding to $C_i$ and by $y_i$ the characteristic vector of $C_i$. By assumption, there is at most one index $i \in [\ell]$ such that $y_i$ is homologous to zero.  Therefore, at most one of the curves $\gamma_1, \dots, \gamma_\ell$ is contractible.  Next, if $\gamma_i$ and $\gamma_j$ are disjoint and freely homotopic, then $\gamma_i \cup \gamma_j$ bounds a cylinder in $\surf$ by a result of Epstein~\cite{Epstein66}. In particular, $y_i$ and $y_j$ are homologous.  Therefore, there is a set of indices $I \subseteq [\ell]$ such that the curves in $\{\gamma_i \mid i \in I\}$ are noncontractible and for all distinct $i,j \in I$, $\gamma_i$ is not freely homotopic to $\gamma_j$ nor $\gamma_j^{-1}$, and $\ell \leq 2|I|+1$. After renumbering, we can assume that $I = \{1,\ldots,|I|\}$. Now apply Lemma~\ref{lem:MM92} to $\gamma_1$, \ldots, $\gamma_{|I|}$.
\end{proof}

Next, we observe that optimal solutions (which we may assume to be vertices of $ \slack(G) $) can be decomposed in the following way.

\begin{lem} \label{lem:decomposition}
    Let $G$ be a graph satisfying the standard assumptions and let $ (D,\omega) $ be a dual representation of $ G $. Every vertex $ y $ of $ \slack(G) $ can be written as
    \[
        y = y_1 + \dotsb + y_\ell
    \]
    with $ \ell = O(g) $, where $ y_1,\dotsc,y_\ell $ are characteristic vectors of strongly connected (directed) Eulerian subgraphs of $ D $ that are non-homologous to $ \zerovec $.
\end{lem}

\begin{proof}
By Proposition~\ref{prop:subG_01-vertices} we have $y \in \{0,1\}^{A(D)}$. Thus we can decompose $y$ as the sum of the characteristic vectors of a certain number of edge-disjoint directed cycles in $D$, say
\[
    y = \sum_{i=1}^K \chi^{C_i}\,.
\]
In what follows, we will make use of the following claim: \textit{There is no non-zero integer circulation $ y' $ in $ D $ with $ \omega(y') = \zerovec $ such that $ y + y' $ and $ y - y' $ are non-negative.} Indeed, otherwise $ y + y' $ and $ y - y' $ are contained in $ \slack(Q) $. This means that $ y $ is a convex combination of two other points in $ Q(D) $, a contradiction to the fact that $ y $ is a vertex.

Note that this implies that $ \omega(\chi^{C_i}) \ne \omega(\chi^{C_j}) $ for each $ i,j \in [K] $, $ i \ne j $, otherwise $ y' := \chi^{C_i} - \chi^{C_j} $ would contradict the above claim.

Let $I$ denote a maximal subset of $[K]$ such that the directed cycles in $\{C_i \mid i \in I\}$ are node-disjoint. By Lemma~\ref{lem:MM92}, we have that $|I| = O(g)$. We can partition the cycles of the decomposition in $|I|$ collections $\mathcal{C}_i$, $i \in I$ such that each cycle in $ \mathcal{C}_i $ has a common node with $ C_i $. After reindexing if necessary, we can assume that $I = \{1,2,\dotsc,|I|\}$. Now we can define $y_i := \sum_{j : C_j \in \mathcal{C}_i} \chi^{C_j}$. Observe that each $ y_i $ is a non-empty strongly connected Eulerian subgraph of $ D $, and that $ y = y_1 + \dotsb + y_{|I|} $.

Finally, we have that no $ y_i $ can be homologous to $ \zerovec $, otherwise $ y' := y_i $ would contradict the above claim.
\end{proof}

Motivated by Lemma~\ref{lem:decomposition}, we will perform a sequence of min-cost flow computations in the following auxiliary graph.

\begin{defn}
    \label{defnCoverGraphAlgorithm}
    Let $ G $ be a graph satisfying the standard assumptions and let $ (D,\omega) $ be a dual representation of $ G $. We define the corresponding \emph{cover graph} $ \bar{D} $ as the directed graph with nodes
    \begin{align*}
        V(\bar{D}) := \{ (f, p) \mid f \in V(D), \, b \in \Z_2 \times \{-|E(G)|, \dotsc, |E(G)|\}^{g-1} \}, \,
    \end{align*}
    and arcs
    \begin{align*}
        A(\bar{D}) := \{ ((f, b), (f', b')) \mid \ & (f, b), (f', b') \in V(D), \, (f, f') \in A(D), \\
        & b + \omega(\chi^{\{(f, f')\}}) = b' \}\,.
    \end{align*}
    Given arc costs $ c : A(D) \to \R $, we define arc costs $ \bar{c} : A(\bar{D}) \to \R $ via
    \[
        \bar{c}((f, b), (f', b')) := c(f, f')\,.
    \]
    Finally, for any $ \bar{y} \in \Z^{A(\bar{D})} $ we define its \emph{projection} to be the vector $ y \in \Z^{A(D)} $ given by
    \[
        y(f, f') := \sum_{b, b' :  ((f, b), (f', b')) \in A(\bar{D})} \bar{y}((f, b), (f', b'))\,,
    \]
    and we say that $ y $ \emph{lifts} to $ \bar{y} $.
\end{defn}

\begin{obs}
    \label{obsProjectionCost}
    If $ y $ is the projection of $ \bar{y} $, then $ c(y) = \bar{c}(\bar{y}) $.
\end{obs}

\begin{obs}
    \label{obsProjection}
    Every non-negative integer unit flow from $ (f, \zerovec) $ to $ (f, b) $ in $ \bar{D} $ projects to a non-negative integer circulation $ y $ in $ D $ with $ \omega(y) = b $.
\end{obs}

\begin{obs}
    \label{obsProjectionOfCirculation}
    Every integer circulation in $ \bar{D} $ projects to an integer circulation $ y $ in $ D $ with $ \omega(y) = \zerovec $.
\end{obs}

\begin{lem}
    \label{lemLift}
    Let $ G $ be a graph satisfying the standard assumptions and let $ (D,\omega) $ be a dual representation of $ G $. Every characteristic vector $ y \in \{0,1\}^{A(D)} $ of a strongly connected Eulerian subgraph of $ D $ with $ \omega(y) \ne \zerovec $ lifts to the characteristic vector of a non-negative integer unit flow from $ (f, \zerovec) $ to $ (f, \omega(y)) $ for some $ f \in V(D) $.
\end{lem}

\begin{proof}
    Let $H \subseteq D$ denote the strongly connected Eulerian subgraph of $D$ such that $y = \chi^H$. Consider an Euler tour for $D$, say $W$. Let $f$ denote the initial node of $W$. Thus $W = (f_0,a_1,f_1,\dotsc,a_k,f_k) $ is a (directed) closed walk in $D$ starting and ending at $f$, i.e., $ f_0 = f_k = f $. Set $ b_0 := \zerovec \in \Z_2 \times \Z^{g-1} $ and for each $ i \in [k] $ define $ b_i := b_{i-1} + \omega(\chi^{\{a_i\}}) $. Since $ H $ consists of at most $ |A(D)| = |E(G)| $ arcs, and since all components of $ \omega(\chi^{\{a_i\}}) $ are in $ \{-1,0,1\} $ for each $ i $, we have that $ b_0, \dotsc, b_k \in \Z_2 \times \{-|E(G)|, \dotsc, |E(G)|\}^{g-1} $. This means that for each $ i $, the node $ (f_i, b_i) $ is contained in $ V(\bar{D}) $. Furthermore, note that
    \[
        b_k = b_0 + \sum_{i=1}^k \omega(\chi^{\{a_i\}}) = \omega(\chi^H).
    \]
    Starting at $ (f_0, b_0) $ and sending one unit along the nodes $ (f_1,b_1), \dotsc, (f_k, b_k) $ we obtain a non-negative integer unit flow from $ (f, \zerovec) $ to $ (f, b_k) = (f, \omega(\chi^H)) $ that projects to $ y = \chi^H $.
\end{proof}

We are ready to state our algorithm\footnote{Actually, it is not hard to check that the complexity of our algorithm for Problem~\ref{probCirculationHomologous} is $n^{O(g^2)}$ where $n$ is the number of nodes in $G$.} to solve Problem~\ref{probCirculationHomologous}. Recall that we are given a graph $ G $ satisfying the standard assumptions, and non-negative edge costs $ c : E(G) \to \Q_{\ge 0} $. By Proposition~\ref{propRepresentationSystem}, we can efficiently compute a dual representation $ (D,\omega) $ of $ G $. Recall that we can interpret $ c $ also as costs on the arcs of $ D $. As a next step, for each node $ f \in V(D) $ and each $ b \in \Z_2 \times \{-|E(G)|, \dotsc, |E(G)|\}^{g-1} $ we compute a minimum $ \bar{c} $-cost non-negative integer unit flow $ \bar{y}_{f,b} $ in $ \bar{D} $ from $ (f, \zerovec) $ to $ (f,b) $. It is well-known that each such computation can be performed in time polynomial in the size of $ \bar{D} $. Regarding $ g $ as a constant, both the size of $ \bar{D} $ as well as the number of relevant pairs $ (f, b) $ is polynomial in the size of $ G $.

Having done these computations, in view of Lemma~\ref{lem:decomposition}, Observation~\ref{obsProjectionCost}, Observation~\ref{obsProjection}, and Lemma~\ref{lemLift} it remains to enumerate all lists of flows $ \bar{y}_{f_1,b_1}, \dotsc, \bar{y}_{f_t,b_t} $ with $ t \le \ell = O(g) $ and $ b_1 + \dotsb + b_t = (1,\zerovec) $. Among those lists, take the one that minimizes $ \bar{c}(\bar{y}_{f_1,b_1}) + \dotsb + \bar{c}(\bar{y}_{f_t,b_t}) $. The sum of the projections of these flows give the desired optimal circulation in $ D $. Note that we can afford to enumerate all above subsets, again since we regard $ g $ as a constant. 

\section{An \EP{} theorem for $2$-sided odd cycles}
\label{secEP}

In this section we prove Theorem~\ref{EPodd2cycles}.  The proof is by induction on the Euler genus. The base cases are the sphere, the projective plane, and the torus.  The sphere case follows from the fact that odd cycles have the \EP{} property in planar graphs, see~\cite{Reed99} and~\cite{FHRV07}. We use the best current bound, due to Kr\'{a}\v{l}, Sereni, and Stacho~\cite{KSS12}. Below, an \emph{odd cycle transversal} of a graph $G$ is a set $X \subseteq V(G)$ such that $G - X$ is bipartite. (Recall that an odd cycle packing in $G$ is a collection of node-disjoint odd cycles.)

\begin{thm}[Kr\'{a}\v{l}, Sereni, and Stacho~\cite{KSS12}] \label{EPplanar}
For every planar graph $G$, there exists an odd cycle transversal $X$ and an odd cycle packing $\mathcal{C}$ such that $|X| \leqslant 6 |\mathcal{C}|$.
\end{thm}

We begin by proving Theorem~\ref{EPodd2cycles} for the projective plane.
Before doing so, we need to establish a few intermediate results.
Let $G$ be a graph and $v \in V(G)$. Let $W=(v_0, e_1, v_1,  \dots e_k, v_k)$ be a closed walk in $G$. Recall that the multiplicity of $v$ in $W$, is the size of the set $\{j \in [k] \mid v_j=v\}$.  Let $\mathcal W$ be a collection of closed walks in $G$. The \emph{total multiplicity} of  $v$ in $\mathcal W$ is the sum of the multiplicities of $v$ over each $W \in \mathcal W$. We say that $\mathcal W$ is a \emph{$2$-packing} if each node in $G$ has total multiplicity at most $2$.

Suppose that $G$ is cellularly embedded in the projective plane $\mathbb{P}$, and $\Sigma$ is the signature of the embedding.
The \emph{planar double cover} of $G$ is the graph having two copies $v^1$ and $v^2$ of each node $v$ of $G$, two edges $u^1v^1$ and $u^2v^2$ for each edge $uv \in E(G) \setminus \Sigma$ and two edges $u^1v^2$ and $u^2v^1$ for each edge $uv \in \Sigma$.

\begin{lem} \label{EP2packing}
For all graphs $G$ embedded on the projective plane, there exists a set $X \subseteq V(G)$ and a $2$-packing $\mathcal{W}$ of $2$-sided odd closed walks such that $G-X$ does not contain a $2$-sided odd cycle and $|X| \leqslant 6|\mathcal W|$.
\end{lem}
\begin{proof}
We may assume that the embedding of $G$ in $\mathcal{P}$ is cellular. Otherwise, $G$ is planar and the result follows directly from Theorem~\ref{EPplanar}. Now, let $G'$ be the planar double cover of $G$. By Theorem~\ref{EPplanar}, there exists an odd cycle transversal $X'$ and an odd cycle packing $\mathcal{C}'$ of $G'$ such that $|X'| \leqslant 6 |\mathcal{C}'|$.  Let $X$ be the set of nodes $x \in V(G)$ such that $x^1 \in X'$ or $x^2 \in X'$.  Each $2$-sided odd cycle $C$ of $G$ lifts to two odd cycles $C^1$ and $C^2$ of $G'$.  Since $X'$ is an odd cycle transversal of $G'$, $X$ contains a node of $C$.  Therefore, $G-X$ does not contain a $2$-sided odd cycle.  Conversely, each odd cycle $C$ of $G'$ projects to a $2$-sided odd closed walk $W_C$ of $G$.  Since the cycles in $\mathcal C'$ are node-disjoint, the collection $\mathcal W:=\{W_C \mid C \in \mathcal C'\}$ is a $2$-packing of $2$-sided odd closed walks, as required.
\end{proof}

A \emph{bicycle} is a closed walk consisting of the union of two $1$-sided cycles $C_1$ and $C_2$ that meet in exactly one node.  Note that a bicycle is $2$-sided.

\begin{lem} \label{EPminimalwalks}
Let $G$ be a graph embedded in the projective plane.  Let $W$ be a $2$-sided odd closed walk in $G$. Then $W$ contains a $2$-sided odd cycle or an odd bicycle.
\end{lem}

\begin{proof}
We view a closed walk as an Eulerian multigraph.
Let $W$ be a counterexample with $|E(W)|$ minimum.
  Since $W$ is Eulerian, $E(W) = E(C_1) \cup \dots \cup E(C_\ell)$, where $C_1, \dots C_\ell$ are edge-disjoint cycles.  By relabelling, we may assume that $C_1, \dots, C_k$ are all $1$-sided and $C_{k+1}, \dots, C_\ell$ are all $2$-sided.  Furthermore, $C_{k+1}, \dots, C_\ell$ are all even, else we are done.  Since $W$ is odd, $\sum_{i \in [k]} |E(C_i)|$ is odd.  In particular, $k \geq 1$.  Since $W$ is $2$-sided, $k \geq 2$ and is even.  By relabelling, we may assume that $C_1 \cup C_2$ is odd.   Since every two $1$-sided curves on the projective plane intersect, $|V(C_1) \cap V(C_2)| \geq 1$.  Therefore, $C_1 \cup C_2$ is a $2$-sided odd closed walk. By minimality, $E(W)=E(C_1) \cup E(C_2)$.
  If $|V(C_1) \cap V(C_2)| \geq 2$, then $W$ contains a $2$-sided cycle $C$.  Let $C \cup C_1' \cup \dots C_m'$ be a decomposition of $E(W)$ into edge-disjoint cycles.   By the previous argument, $m=1$ and $C$ and $C_1'$ are both $1$-sided, which is a contradiction. Therefore, $|V(C_1) \cap V(C_2)|=1$, and $W$ is an odd bicycle, as required.
\end{proof}

 Let $G$ be a graph embedded on the projective plane.  A \emph{painted walk} in $G$ is a pair $(W, \Delta)$, where $W$ is a closed walk and $\Delta$ is a component of $\mathbb{P} \setminus W$ which is an open disk.
 Note that there is a unique choice for $\Delta$ if $W$ is a $2$-sided cycle and two choices for $\Delta$ if $W$ is a bicycle. Two open disks $\Delta_1$ and $\Delta_2$ \emph{cross} if $\Delta_1 \cap \Delta_2$, $\Delta_1 \setminus \Delta_2$, and $\Delta_2 \setminus \Delta_1$ are all non-empty. Let $\mathcal W$ be a collection of painted walks in $G$.
 We say that $\mathcal W$ is \emph{laminar} if for all $(W_1, \Delta_1), (W_2, \Delta_2) \in \mathcal W$, $\Delta_1$ and $\Delta_2$ do not cross.
 We let $\overline{\mathcal{W}}$ be the collection of underlying walks in $\mathcal W$ and we call $\mathcal W$ a \emph{$2$-packing}, if $\overline{\mathcal{W}}$ is a $2$-packing.

 \begin{lem} \label{EPtwo_bicycles}
 Let $(W_1, \Delta_1)$ and $(W_2, \Delta_2)$ be painted bicycles such that $\{(W_1, \Delta_1), (W_2, \Delta_2)\}$ is a $2$-packing.  Then $\Delta_1$ and $\Delta_2$ cross.
 \end{lem}

 \begin{proof}
 Suppose not.  Let $v$ be the degree-$4$ node of $W_1$.  By symmetry we may assume $\Delta_1 \subseteq \Delta_2$ or $\Delta_1 \cap \Delta_2=\emptyset$.  In either case, $W_2 \subseteq \mathbb{P} \setminus (\Delta_1 \cup \{v\})$.  However, this is impossible since $\mathbb{P} \setminus (\Delta_1 \cup \{v\})$ does not contain any $1$-sided curves.
 \end{proof}

\begin{lem} \label{EPuncrossing}
Let $G$ be a projective planar graph. Suppose $G$ contains a $2$-packing $\mathcal W$ of $k$ painted walks, such that each walk in $\overline {\mathcal W}$ is a $2$-sided odd cycle or an odd bicycle.  Then, $G$ contains a laminar $2$-packing $\mathcal W'$ of $k-1$ painted walks such that each walk in $\overline{\mathcal W}'$ is a $2$-sided odd cycle.
\end{lem}

\begin{proof}
If $(W, \Delta)$ is a painted walk, we let $a(W, \Delta)$ be the area of $\Delta$. If  $\mathcal W$ is a collection of painted walks, we let
\[
\nabla \mathcal W:=\left(\sum_{(W, \Delta) \in \mathcal W} |E(W)|, \sum_{(W, \Delta) \in \mathcal W} a(W,
\Delta), -\sum_{(W, \Delta) \in \mathcal W} a(W, \Delta)^2\right).
\]
We choose $\mathcal W$ such that $\nabla \mathcal W$ is lexicographically minimal among all $2$-packings $\mathcal W'$ of $k$ painted walks, such that each walk in $\mathcal W'$ is a $2$-sided odd cycle or an odd bicycle.   We claim that such a $\mathcal W$ is laminar.  Suppose not. Let $(W_1, \Delta_1), (W_2, \Delta_2) \in \mathcal W$ be such that  $\Delta_1$ and $\Delta_2$ cross, and $a(W_1, \Delta_1) \leq a(W_2, \Delta_2)$. Let $\mathsf{cl} (\Delta_1)$ be the closure of $\Delta_1$, and $H=W_2 \cap \mathsf{cl}(\Delta_1)$.
Since $\Delta_1$ and $\Delta_2$ cross and $\mathcal W$ is a $2$-packing, $H$ contains a path $P$ such that both ends of $P$ are on $W_1$, but no other node of $P$ is on $W_1$.   Observe that there are edge-disjoint paths $P_1$ and $P_2$ in $W_1$ such that $E(W_1)=E(P_1) \cup E(P_2)$ and each of $P \cup P_1$ and $P \cup P_2$ is a $2$-sided cycle or a bicycle.  Since $W_1$ is odd, exactly one of $P \cup P_1$ or $P \cup P_2$ is odd. By symmetry, suppose $W_1':=P \cup P_1$ is odd. Note that $W_1'$ bounds a disk $\Delta_1'$ strictly contained in $\Delta_1$.  Furthermore, $W_2 \setminus E(W_1')$ is a $2$-sided odd closed walk.  By Lemma~\ref{EPminimalwalks}, $W_2 \setminus E(W_1')$ contains a $2$-sided odd cycle or an odd bicycle, $W_2'$.  By minimality, $W_2'=W_2 \setminus E(W_1')$. Therefore, there is a disk $\Delta_2'$ such that $(W_2', \Delta_2')$ is a painted walk and $a(W_1', \Delta_1')+a(W_2', \Delta_2') \leq a(W_1, \Delta_1)+a(W_2, \Delta_2)$.  Letting $\mathcal W':=(\mathcal W \setminus \{(W_1, \Delta_1), (W_2, \Delta_2)\}) \cup \{(W_1', \Delta_1'), (W_2', \Delta_2')\}$, we have that $\nabla \mathcal W'$ is lexicographically smaller than $\nabla \mathcal W$, which is a contradiction.  Thus, $\mathcal W$ is laminar, as claimed.  By Lemma~\ref{EPtwo_bicycles}, $\mathcal{W}$ contains at most one painted bicycle.  Therefore, the lemma follows by removing the painted bicycle from $\mathcal W$, if necessary.
\end{proof}

\begin{lem} \label{EPlaminar}
Let $G$ be a graph embedded on the projective plane $\mathbb{P}$, and let $\mathcal{F}$ be a laminar family of $k$ open disks in $\mathbb{P}$. Suppose that each $\Delta \in \mathcal{F}$ is bounded by a $2$-sided odd cycle $\partial \Delta \subseteq G$, and that $\{\partial \Delta \mid \Delta \in \mathcal{F}\}$ is a $2$-packing. Then $G$ has a packing of $k/18$ $2$-sided odd cycles.
\end{lem}

\begin{proof}
Let $\Delta_1$, \ldots, $\Delta_t$ denote the maximal elements of $\mathcal{F}$. For each $i \in [t]$, let $k_i$ denote the number of elements in the subfamily $\mathcal{F}_i := \{\Delta \in \mathcal{F} \mid \Delta \subseteq \Delta_i\}$. Thus $k = k_1 + \cdots + k_t$. For each $i \in [t]$, let $G_i$ denote the union of the cycles $\partial \Delta$ with $\Delta \in \mathcal{F}_i$ and $\Delta \neq \Delta_i$. Notice that $G_i$ is a planar graph.

Within each $G_i$, the odd cycle transversal number is at least $(k_i - 1)/2$ and thus the odd cycle packing number at least $(k_i - 1)/12$. Pick a maximum odd cycle packing inside each $G_i$. The union of these packings is an odd cycle packing in $G$ of size a least
$$
\sum_{i=1}^t \frac{k_i - 1}{12}
= \frac{k-t}{12}\,.
$$
If $t \leqslant k/3$ then this packing has size at least $k/18$ and we are done.

From now on, assume $t \geqslant k/3$. Let $H$ be the auxiliary graph that has one node $v_i$ for each $\Delta_i$, $i \in [t]$ such that $v_i$ and $v_j$ are adjacent if $\partial \Delta_i$ and $\partial \Delta_j$ have a common node. Since $\Delta_1$, \ldots, $\Delta_t$ are disjoint open disks, $H$ can be embedded in the projective plane. By the Heawood bound, $\chi(H) \leqslant 6$. Therefore, $H$ has a stable set of size at least $t/6$. This stable set yields a packing of $2$-sided odd cycles of size at least $t/6 \geqslant k/18$.
\end{proof}

We are now in position to prove Theorem~\ref{EPodd2cycles} for the projective plane.

\begin{thm} \label{EPprojectiveplane}
For every graph $G$ embedded on the projective plane, there exists a set $X \subseteq V(G)$ such that $G-X$ does not contain a $2$-sided odd cycle, and a collection $\mathcal{C}$ of node-disjoint $2$-sided odd cycles such that $|X| \leqslant 114 |\mathcal{C}|$.
\end{thm}

\begin{proof}
By Lemmas~\ref{EP2packing} and \ref{EPuncrossing}, there exists a set $X \subseteq V(G)$ meeting every $2$-sided odd cycle of $G$ and a laminar $2$-packing $\mathcal{W}'$ of painted $2$-sided odd cycles such that $|X| \leqslant 6 |\mathcal{W}'| + 6$.

By Lemma~\ref{EPlaminar}, we may extract from $\mathcal{W}'$ a packing $\mathcal{C}$ of $2$-sided odd cycles such that $|\mathcal{C}| \geqslant (|\mathcal{W}'| - 1)/18$. Since we may clearly assume that $G$ has a $2$-sided odd cycle (otherwise taking $X := \emptyset$ would make the statement true), we may also assume that $|\mathcal{C}| \geqslant 1$.

For this transversal $X$ and packing $\mathcal{C}$ of $2$-sided odd cycles, we get
$$
|X| \leqslant 6 |\mathcal{W}'| + 6 \leqslant 108 |\mathcal{C}| + 6 \leqslant 114 |\mathcal{C}|\,. \qedhere
$$
\end{proof}

Consider a graph $G$ embedded in a surface $\surf$. A cycle $C$ of $G$ is said to be \emph{noncontractible} if $C$ is noncontractible (as a closed curve) in $\surf$. Also, $C$ is called {\em surface separating} if $C$ separates $\surf$ in two connected pieces. If $\surf$ is a surface other than the sphere, then every contractible cycle $C$ of $G$ bounds exactly one closed disc in $\surf$.

Given a cycle $C$ of a graph, a \emph{$C$-ear} (or simply \emph{ear} when $C$ is clear from
the context) is a path $P$ such that the two ends of $P$ are in $C$ but
no internal node nor any edge of $P$ is in $C$.
If $C$ has even length and $C \cup P$ contains an odd cycle, we say that the ear $P$ is \emph{parity breaking}.

\begin{lem}
\label{lem:ears}
If an odd cycle $C'$ has at least two nodes in common with an even cycle $C$, then
$C'$ contains a parity-breaking $C$-ear. 
\end{lem}
\begin{proof}
Let $G$ be the graph consisting of the union of $C$ and $C'$.
Let $uv \in E(C') \setminus E(C)$.
Let $u'$ be the node from $C$ that is closest to $u$ on the path $C' - v$
(possibly $u'=u$).
Similarly, let $v'$ be the node from $C$ that is closest to $v$ on the path $C' - u$.
Then $u' \neq v'$, since otherwise $V(C) \cap V(C') = \{u'\}$, contradicting
the fact that $C$ and $C'$ have at least two nodes in common. 
Thus the path $P_{uv}$ from $u'$ to $v'$ in $C'$ that includes the edge $uv$ is a $C$-ear.

Since taking the union of $C$ with the paths $P_{uv}$ for all edges $uv \in E(C') \setminus E(C)$
gives the non-bipartite graph $G$, it cannot be that $C \cup P_{uv}$ is bipartite for
each $uv \in E(C') \setminus E(C)$. Hence there exists $uv \in E(C') \setminus E(C)$ such that
$P_{uv}$ is a parity-breaking $C$-ear.
\end{proof}

We will need the following two theorems about toroidal graphs. 

\begin{thm}[Schrijver~\cite{S93}]
\label{thm:Schrijver}
Every graph embedded in the torus with facewidth $t$ contains $\lfloor 3t/4 \rfloor$ node-disjoint noncontactible cycles.
\end{thm}

For an integer $t \geq 3$, the graph $C_{t} \times C_{t}$ denotes the $t\times t$ toroidal grid, namely
the graph with node set $\{(i,j): 0 \leq i,j \leq k-1\}$ and where $(i,j)$ is adjacent to
$(i'j')$ if and only if $i=i'$ and $j \equiv j' \pm 1  \pmod{t}$,
or $j=j'$ and $i \equiv i' \pm 1 \pmod t$.

\begin{thm}[{d}e~Graaf and Schrijver~\cite{dGS}]
\label{thm:de_Graad-Schrijver}
For every integer $t \geq 3$, every graph embedded in the torus with facewidth at least $\frac{3}{2}t$ has
 $C_{t} \times C_{t}$ as a surface minor. 
\end{thm}

In the above theorem, a {\em surface minor} of a graph embedded in a surface is any embedded graph obtained by applying edge deletion, edge contraction, and node deletion operations directly on the embedding. 

Kawarabayashi and Nakamoto~\cite{KN_DM} proved a bound of $O(k)$ on the minimum number of {\em edges} to remove from a toroidal graph having no $k+1$ edge-disjoint odd cycles to make it bipartite. 
For our purposes, we need a `node' version of their result, which fortunately can be obtained by an easy modification of their proof. 

\begin{thm}
\label{th:torus_node}
Every toroidal graph $G$ having no $k+1$ node-disjoint odd cycles has a node subset $X$ of size $|X|\leq 38k + 7$ meeting all odd cycles of $G$. 
\end{thm}
\begin{proof}
Let $G$ be a graph embedded in the torus with no $k+1$ node-disjoint odd cycles. 
Let $t$ denote the facewidth of the embedding. 

{\bf Case 1: $t \leq 32k + 7$.} 
By definition of $t$, there is a noncontractible simple closed curve $\gamma$ in the torus that intersects $G$ in a set $Z$ of exactly $t$ nodes. 
Furthermore, the graph $G-Z$ is planar. 
By Theorem~\ref{EPplanar}, there is a subset $Y$ of nodes of $G-Z$ of size at most $6k$ meeting all odd cycles of $G-Z$. 
Then $X:=Y \cup Z$ meets all odd cycles of $G$ and has size 
\[
|X| \leq 6k + 32k + 7 = 38k + 7, 
\]
as desired. 

{\bf Case 2: $t \geq 32k + 8$.}
Since $\lfloor 3t/4 \rfloor \geq 24k + 6$, 
by Theorem~\ref{thm:Schrijver} the graph $G$ has $24k + 6$ node-disjoint noncontractible cycles. 
These cycles are necessarily pairwise homotopic on the torus, thus they define a cylinder $D$.  
Denote the cycles by $C_{1}, C_{2}, \dots, C_{24k + 6}$, in
an order consistent with $D$ (thus $C_1$ and $C_{24k + 6}$ are on the boundary of the cylinder).

The graphs $H_{1}:=G - V(C_{1})$ and $H_{2}:=G - V(C_{12k + 3})$ are both planar,
hence they each can be made bipartite by removing at most $6k$ nodes by Theorem~\ref{EPplanar}. 
Let thus $X_{i}\subseteq V(H_{i})$ ($i=1,2$) be such that $H_{i} - X_{i}$ is bipartite and $|X_{i}| \leq 6k$.

Let $G':= G - (X_{1} \cup X_{2})$. 
If $G'$ is planar then it suffices to take a subset $Y$ of at most $6k$ nodes of $G'$ meeting all its odd cycles, using Theorem~\ref{EPplanar}. 
Then, $X:=X_{1} \cup X_{2} \cup Y$ meets all odd cycles of $G$ and has size at most $18k$, as desired. 

Thus we may assume that $G'$ is not planar. 
Recalling that Euler genus is additive on $0$- and $1$-sums, it follows that there is a unique block $B$ of $G'$ which is not planar. 
Let $G_0$ be the subgraph of $G'$ obtained by taking the union of all blocks of $G'$ distinct from $B$. 
Note that $G_0$ is planar. 
(Also, note that $G_0$ is not necessarily connected, though this is not important for our purposes.) 
Let $Y$ be a subset at most $6k$ nodes of $G_0$ meeting all odd cycles of $G_0$, using Theorem~\ref{EPplanar} once more, and let $X:=X_{1} \cup X_{2} \cup Y$. 
Since $|X| \leq 18k$, it suffices to show that $X$ meets all odd cycles of $G$. 
Clearly, if there is any odd cycle left in $G-X$, it must be contained in the toroidal block $B$ of $G'$. 
Hence, in order to finish the proof, it only remains to show that $B$ is bipartite. 

Let us consider the toroidal embedding of $B$ induced by that of $G$. 
Before showing that $B$ is bipartite, we remark that each face of $B$ is bounded by a cycle of $B$, since $B$ is $2$-connected and cellularly embedded with facewidth at least $t - |X_{1} \cup X_{2}| \geq 2$. 
Furthermore, the embedding of $B$ must be {\em even}, meaning that all these facial cycles are even, for the following reason. 
Let $C$ be a facial cycle of $B$. 
Since $|X_{1} \cup X_{2}| \leq 12k$ there are indices $i, j$ with $1 < i < 12k + 3 < j < 24k + 6$ 
such that $X_{1} \cup X_{2}$ avoids both $C_i$ and $C_j$. 
Consider the two cylinders on the torus bounded by $C_i$ and $C_j$. 
The facial cycle $C$ must be contained in one of these two cylinders. 
However, in the graph $G$, one of the two cycles $C_1, C_{12k + 3}$ is drawn completely outside that cylinder, implying that $C$ is a subgraph of $H_i-X_i$ for some $i\in\{1,2\}$. 
Hence, $C$ is even, as claimed. 

The fact that the embedding of $B$ is even implies the following property: 
\[
\textrm{(*) Every two cycles in $B$ in the same $\Z_2$-homology class have the same parity.}
\]
To see this, recall that two cycles $C_1, C_2$ in $B$ are in the same $\Z_2$-homology class if and only if $C_2=C_1+\sum_{F \in \mathcal{F}} F$, where $\mathcal F$ is a set of facial cycles and computations are performed in $\Z_2^{E(G)}$ (for the sake of simplicity, we identify cycles with their characteristic vectors). Since all cycles in $\mathcal F$ are even, the claim follows.  

Next, we show that $B$ must be bipartite.  
Let $t'$ denote the facewidth of the embedding of $B$. 
Note that $t' \geq t - |X_{1} \cup X_{2}| \geq \frac{3}{2}(k+1)$. 
Hence, $C_{k+1} \times C_{k+1}$ is a surface minor of $B$, by Theorem~\ref{thm:de_Graad-Schrijver}. 
In the drawing of $C_{k+1} \times C_{k+1}$, consider the set of $k+1$ node-disjoint noncontractible and pairwise homotopic cycles that are `horizontal', and the set of those that are `vertical'. 
For each such set, there is a corresponding set of $k+1$ node-disjoint homotopic noncontractible cycles in $B$. 
By property (*), it follows that these two sets of $k+1$ node-disjoint cycles of $B$ are all even, since $G$ has no $k+1$ node-disjoint odd cycles. 
In particular, $B$ contains two even cycles $C_1$ and $C_2$ that generate the $\Z_2$-homology group of the torus. Since the cycle space of $B$ is generated by $C_1, C_2$, and the facial cycles of $B$, and all these cycles are even, it follows that every cycle of $B$ is even.  
\end{proof}

We will use the following result of Brunet, Mohar, and Richter~\cite{BMR_JCTB} (see Corollary 7.3 and Corollary 8.1.1 of that paper).

\begin{thm}[Brunet, Mohar, and Richter~\cite{BMR_JCTB}]
\label{thm:BMR}
Let $G$ be a graph embedded in a surface $\surf$ with facewidth $t \geq 3$. 
Then the following are sufficient conditions for $G$ to contain $k$ node-disjoint surface separating cycles that are noncontractible and pairwise homotopic:
\begin{itemize}
    \item $\surf$ is orientable with Euler genus at least $4$ and $\frac{t - 9}{8} \geq k$;
    \item $\surf$ is nonorientable with Euler genus at least $2$ and $\frac{t - 1}{4} \geq k$.
\end{itemize}
\end{thm}

A \emph{surface with $h$ holes} $\surf$ is obtained by removing $h$ open disks with pairwise disjoint closures from a surface $\surf'$.  The \emph{Euler genus} of $\surf$ is the Euler genus of $\surf'$.  Each component of the boundary of $\surf$ is called a \emph{cuff}.  Let $G$ be a graph embedded in a surface with holes $\surf$.  A \emph{pattern} in $G$ is a collection $\Pi=\{(s_1, t_1), \dots, (s_k, t_k)\}$ of pairs of vertices of $G$ where $s_1, \dots, s_k, t_1, \dots, t_k$ are all distinct.  We let $\overline{\Pi}=\{s_1, \dots, s_k, t_1, \dots, t_k\}$.  A \emph{$\Pi$-linkage} in $G$ is a collection of vertex-disjoint paths $P_1, \dots, P_k$ such that $P_i$ has ends $s_i$ and $t_i$ for all $i \in [k]$.
We say that $\Pi$ is \emph{topologically feasible} in $\surf$, if there is a collection of disjoint curves $\gamma_1, \dots, \gamma_k$ in $\surf$ such that $\gamma_i$ connects $s_i$ and $t_i$ for all $i \in [k]$.  We require two theorems of Robertson and Seymour~\cite{RS88} giving sufficient conditions for the existence of $\Pi$-linkages in graphs embedded in surfaces with holes.

The first gives sufficient conditions for linkages across a cylinder~\cite[(5.8)]{RS88}.

\begin{thm} \label{thm:cylinderlinkage}
Let $\surf$ be the cylinder and $O_1$ and $O_2$ be its two cuffs.  Let $G$ be embedded in $\surf$ such that $G \cap (O_1 \cup O_2) \subseteq V(G)$.  If $G$ has facewidth at least $t$ and every $G$-normal curve $\gamma$ in $\surf$ between $O_1$ and $O_2$ has length at least $t$, then there are pairwise node-disjoint noncontractible cycles $C_1, \dots C_t$ of $G$ and pairwise node-disjoint paths $P_1, \dots, P_t$ of $G$ such that each path has one end on $O_1$ and the other on $O_2$, and the intersection of each path with each cycle is a path.  
\end{thm}

The second is a simplified version of~\cite[(7.5)]{RS88}.

\begin{thm}[Robertson and Seymour~\cite{RS88}]
\label{thm:linkage}
There exists a computable function $b(g, k)$ such that the following holds. 
Let $\surf$ be a surface with $1$ hole and Euler genus $g \neq 0$.  Let $O$ be the unique cuff of $\surf$.  
Let $G$ be a graph embedded in $\surf$ and $\Pi=\{(s_1, t_1), \dots, (s_k, t_k)\}$ be a pattern in $G$ with $\overline{\Pi} \subseteq G \cap O$.  If $\Pi$ is topologically feasible in $\surf$, the facewidth of $G$ is at least $b(g, k)$, and every $G$-normal curve between two vertices of $\overline{\Pi}$ has length at least $b(g, k)$, then there is a $\Pi$-linkage in $G$.  
\end{thm}

Note that the computability of the function $b(g,k)$ in Theorem~\ref{thm:linkage} follows from later work of Geelen, Huynh, and Richter~\cite{GHR18} and Matou\v{s}ek, Sedgwick, Tancer, and Wagner~\cite{MSTW16}.\footnote{One can simply follow the proof of Robertson and Seymour, but instead use a computable version of the function $w(\Sigma, k, n)$ from $(3.6)$ of~\cite{RS88}, which follows from~\cite{GHR18} or~\cite{MSTW16}.  In fact, $w(\Sigma, k, n)$ turns out to be independent of the surface $\Sigma$.}
Finally, we will need the following lemma. 

\begin{lem} \label{lem:equivalent_signatures}
    Let $G$ be a $2$-connected graph embedded in a surface in such a way that $G$ has no $2$-sided odd cycle. If $G$ is non-bipartite, then $G$ has no $1$-sided even cycle. That is, a cycle of $G$ is odd if and only if it is $1$-sided. Moreover, $G$ contains no $2$-sided odd closed walk. In other words, the embedding of $G$ is parity-consistent.
\end{lem}

\begin{proof}
     Let $\Sigma$ denote any signature for the embedding. Thus every odd cycle in $G$ is $\Sigma$-odd. We prove, similarly as in Lemma~\ref{lemAuxSignatureIsWholeEdgeSet}, that one can find nodes $v_1$, \ldots, $v_k$ such that $\Sigma \triangle \delta(v_1) \triangle \dotsb \triangle \delta(v_k) = E(G)$.

     Let $ C_0 $ be an odd cycle and $ C_0 \cup P_1 \cup \dots \cup P_\ell $ an ear-decomposition of $ G $. We proceed by induction over $ \ell \geqslant 0 $. In the case $ \ell = 0 $, $ G $ only consists of the cycle $ C_0 $ and one can easily find nodes $ v_1,\dotsc,v_k \in V(C_0) $ such that $ \Sigma' := \Sigma \triangle \delta(v_1) \triangle \dotsb \triangle \delta(v_k) $ consists of all edges of $C_0$ except possibly one. Since $ C_0 $ is odd, it is $\Sigma$-odd and thus $\Sigma'$-odd. This means that all edges of $C_0$ are in $\Sigma'$.

     Now, assume $ \ell \ge 1 $ and consider $ G' := C_0 \cup P_1 \cup \dots \cup P_{\ell-1} $ together with the induced signature $ \Sigma \cap E(G') $. By the induction hypothesis there exist nodes $ v_1,\dotsc,v_{k'} \in V(G') $ such that $ \Sigma \triangle \delta(v_1) \triangle \dotsb \triangle \delta(v_{k'}) $ contains $ E(G') $.

     Taking the symmetric difference with further nodes on $ P_\ell $, we find nodes $ v_{k'+1}, \dotsc, v_k $ such that all edges of $G$ except possibly one edge on the path $P_\ell$ are in $ \Sigma' := \Sigma \triangle \delta(v_1) \triangle \dotsb \triangle \delta(v_{k}) $. Let $e \in E(P_\ell)$ denote the edge of $G$ that could possibly not belong to $\Sigma'$.  Since $ G $ is $ 2 $-connected and non-bipartite, we can find an odd cycle $ C $ that contains $ e $. We proceed as in the case $ \ell = 0 $ and observe that $C$ is $\Sigma'$-odd, meaning that $ \Sigma' $ contains all edges of $ C $ and hence $ \Sigma' = E(G) $.

     By switching the local orientation at $v_1$, \ldots, $v_k$, we may assume that the signature $\Sigma$ equals $E(G)$. It directly follows that every even cycle is $2$-sided. In order to prove the signature-consistency of the embedding, we show by induction on the length of a closed walk $W$ that $W$ is either $1$-sided and odd, or $2$-sided and even.
     
     A \emph{subwalk} of a walk $W=(v_0, e_1, v_1, \ldots, e_k, v_k)$ is any walk of the form $W' := (v_i,e_{i+1},v_{i+1},\ldots,e_j,v_j)$ where $0 \leqslant i \leqslant j \leqslant k$. If $W'$ is closed (that is, $v_i=v_j$), we can obtain a new walk $W - W' := (v_0, e_1, v_1, \ldots, e_{i},v_i,e_{j+1},v_{j+1},\ldots,e_k,v_k)$.

     Consider a simple closed subwalk $W'$ of $W$. Thus $W'$ either is a cycle or $W' = (u,e,v,e,u)$ for some edge $e = uv$. Letting $W'' := W - W'$, we may apply the induction hypothesis to $W''$ to infer that the closed walk $W''$ is either $1$-sided and odd, or $2$-sided and even. Since $W'$ is also either $1$-sided and odd, or $2$-sided and even, this implies that $W$ is $1$-sided and odd, or $2$-sided and even.
\end{proof}

We may now prove Theorem~\ref{EPodd2cycles}, which follows from the following theorem. For a simple closed curve $\gamma$ in a surface $\surf$, we let $\surf \cut \gamma$ be the surface(s) obtained by cutting $\surf$ open along $\gamma$ and then gluing a disk onto each resulting hole.

\begin{thm}
\label{thm:EP-2sided}
There exists a computable function $f:\N \times \N \to \N$ such that the following holds. 
Let $G$ be a cellularly embedded graph in a surface $\surf$ of Euler genus $g$ such that $G$ has no $k+1$ node-disjoint $2$-sided odd cycles. Then there exists a subset $X$ of nodes of $G$ with $|X|\leq f(g, k)$ such that $X$ meets all $2$-sided odd cycles of $G$. 
Furthermore, if $\surf$ is orientable, then there is such a set $X$ of size at most $19^{g+1} \cdot k$. 
\end{thm}

\begin{proof}
Let $c:\N\times \N \to \N$ be the function $c(g, k)=2(b(g,k)+1)^2 + 1$, where $b(g,k)$ is the function from Theorem~\ref{thm:linkage}.  
We prove the theorem with the function $f$ defined inductively on $g$ by setting for every $k\geq 0$
\begin{align*}
f(0,k) &= 6k    \\
f(1,k) &= 114k \\
f(g,k) &= 9m(g,k) + 8c(g,k) + 8 \quad \quad \forall g \geq 2,  
\end{align*}
where 
\[
m(g,k)=\max_{i\in [g-1]}\left( f(i,k) + f(g-i,k) \right).
\]

We prove the theorem with the function $f(g,k)$ defined above. 
We postpone the argument for the linear upper bound in the orientable case until the end of the proof. 

The proof is by induction on $g$. 
Clearly, we may assume $k\geq 1$, since otherwise the theorem trivially holds. 
We treat the sphere, the projective plane, and the torus as base cases.
These cases are covered respectively by Theorem~\ref{EPplanar}, Theorem~\ref{EPprojectiveplane}, and Theorem~\ref{th:torus_node}. 
(For the torus, we use that $38k + 7 \leq f(1, k) \leq f(2, k)$.) 

Next we consider the inductive case.
Thus our surface $\surf$ is either an orientable surface with Euler genus at least $4$, or a nonorientable one with Euler genus at least $2$.

Let $t$ denote the facewidth of the embedding of $G$.
We distinguish two cases, depending on whether the facewidth is small or big compared to $g$ and $k$.

{\bf Case 1: $t \leq 8m(g,k) + 8c(g,k) + 8$.}
By definition of $t$,
there is a noncontractible simple closed curve $\gamma$ in $\surf$ that
intersects $G$ in a set $Z$ of exactly $t$ nodes.

{\bf Case 1.1: $\gamma$ separates $\surf$}. 
In this case, $\surf \cut \gamma$ is the disjoint union of two surfaces $\surf_{1}$,
$\surf_{2}$ of Euler genus $g_{1}$ and $g_{2}$, respectively, such that
$g_{1}, g_{2} \geq 1$ and $g_{1} + g_{2} = g$.
It follows that $G - Z$ is the disjoint union of two graphs $G_{1}$ and $G_{2}$ with $G_{i}$ cellularly embedded in $\surf_{i}$ for each $i \in [2]$. Note that possibly $V(G_{i}) = \varnothing$.
Observe that a cycle of $G_i$ is $2$-sided in the induced embedding of $G_i$ if and only if it is $2$-sided in the embedding of $G$.
Thus, by induction, there is a subset $X_{i}$ of nodes of $G_{i}$ with $|X_i| \leq f(g_i,k)$ meeting all $2$-sided odd cycles of $G_i$, or equivalently, meeting all $2$-sided odd cycles of $G$ that are contained in $G_i$ (as subgraphs).

Let $X:= X_{1} \cup X_{2} \cup Z$.
Clearly, $X$ meets all $2$-sided odd cycles of $G$, and 
\[
|X| \leq f(g_1,k) + f(g_2,k) + t  
\leq m(g,k) +  8m(g,k) + 8c(g,k) + 8 
= f(g,k),
\]
as desired.

{\bf Case 1.2: $\gamma$ does not separate $\surf$}.
In this case, $\surf':=\surf \cut \gamma$ is a surface of Euler genus $g-1$ if $\gamma$ is $1$-sided, or a surface of Euler genus $g-2$ if $\gamma$ is $2$-sided.  
The argument is the same as in the previous case, except we only need to consider one surface of smaller Euler genus.

Again, a cycle of $G-Z$ is $2$-sided in the induced (cellular) embedding of $G-Z$ in $\surf'$ if and only if the cycle is $2$-sided in the embedding of $G$ in $\surf$.
By induction, there is a node subset $X'$ of $G - Z$ of size at most $f(g-1,k)$ meeting all $2$-sided odd cycles of $G-Z$.
Let $X:= X' \cup Z$.
Then $X$ meets all $2$-sided odd cycles of $G$, and
\[
|X| \leq f(g-1,k)  + t
\leq m(g,k) +  8m(g,k) + 8c(g,k) + 8 
= f(g,k),
\]
as desired.

{\bf Case 2:  $t \geq 8m(g,k) + 8c(g,k) + 9$.}
Let
$$
q:= 
\max\left\{\left\lfloor \frac{t - 9}{8} \right\rfloor, 
\left\lfloor \frac{t - 1}{4} \right\rfloor\right\}
=\left\lfloor \frac{t - 9}{8} \right\rfloor
\geq m(g,k) + c(g,k). 
$$

By Theorem~\ref{thm:BMR}, $G$ contains $q$
node-disjoint surface separating cycles that are noncontractible and pairwise homotopic. 
These cycles define a cylinder $D \subset \surf$; denote these cycles by $C_1, C_2, ..., C_q$, in
an order consistent with $D$ (so $C_1$ and $C_q$ are on the
boundary of the cylinder). In what follows, we will use the fact that each cycle $C_{i}$ separates
the surface $\surf$ into two pieces.

Let $G_{0}$ be the subgraph of $G$ lying on the cylinder $D$.
Let $\surf_{1}$ be the component of $\surf \cut C_1$ which contains $D$ and let $\surf_2$ be the component of $\surf \cut C_q$ which contains $D$.
For $i \in [2]$, let $g_{i}$ be the genus of $\surf_{i}$. Thus,
$g_{1}, g_{2} \geq 1$ and $g_{1} + g_{2} = g$.
Let $G_{i}$ be the subgraph of $G$ contained in $\surf_{i}$.
We have $G_{1} \cup G_{2} = G$ and $G_{1} \cap G_{2} = G_{0}$.  Since $D$ is orientable, there is a signature $\Sigma$ of the embedding of $G$ in $\surf$ such that $\Sigma \cap E(G_0)=\emptyset$.  

Again by induction, for each $i \in [2]$, there is a node subset $X_{i}$ of $G_{i}$ with $|X_i| \leq f(g_i,k)$ meeting all $2$-sided odd cycles of $G_i$, or equivalently, meeting all $2$-sided odd cycles of $G$ that are contained in $G_i$ (as subgraphs). 

Let $H:=G-(X_1 \cup X_2)$. 
Since $|X_1 \cup X_2| \leq f(g_1, k) + f(g_2,k)  \leq f(g, k)$, it is enough to show that $H$ has no $2$-sided odd cycle. 
We argue by contradiction, and suppose that $H$ does have one such cycle.

A cycle $C$ of $G$ is {\em cylinder-essential} if $C$ is drawn on the cylinder $D$ and $C$ is homotopic to $C_1$. 
Let $C$ be cylinder-essential cycle of $H$ distinct from $C_1$ and $C_q$. We let $\surf^{-1}_C$ be the component of $\surf \cut C$ containing the cycle $C_1$ and $\surf^{1}_C$ be the other component of $\surf \cut C$.  For each $\ell \in \{-1,1\}$, we let $H^{\ell}_C$ be the subgraph of $H$ contained in $\surf^{\ell}_C$. Note that for each $\ell \in \{-1,1\}$, there is a copy of $C$ in $H^{\ell}_C$.  
We will show: 
\begin{equation}
\label{eq:prop1}
\textrm{
\begin{minipage}{0.92\textwidth}
For all cylinder-essential cycles $C$ of $H$ and all $\ell \in \{-1,1\}$,  either all $\Sigma$-odd $C$-ears in $H^{\ell}_C$ are parity-breaking, or none of them are. 
\end{minipage}
}
\end{equation}
To see this, first observe that $H^{\ell}_C$ has no $\Sigma$-even odd cycle, since $H^{\ell}_C$ is a subgraph of $G_\ell-X_\ell$. 
Let $B$ be the block of $H^{\ell}_C$ containing $C$. 
Now, each $\Sigma$-odd $C$-ear can be extended via a subpath of $C$ to a $\Sigma$-odd cycle, since $\Sigma \cap E(C) = \emptyset$.  
By Lemma~\ref{lem:equivalent_signatures}, all $\Sigma$-odd cycles contained in $B$ have the same length parity, and so property~\eqref{eq:prop1} follows. 

Next we prove:
\begin{equation}
\label{eq:prop2}
\textrm{
\begin{minipage}{0.92\textwidth}
Let $C$ be a cylinder-essential cycle of $H$. 
Then we can write $\{-1,1\}=\{\ell, \ell'\}$ in such a way that there is a $\Sigma$-odd parity-breaking $C$-ear in $H^{\ell}_C$, and a $\Sigma$-odd {\em non} parity-breaking $C$-ear in $H^{\ell'}_C$. 
\end{minipage}
}
\end{equation}
To see this, consider a $2$-sided odd cycle $C'$ in $H$ (whose existence we assumed). 
Let $\mathcal{P}$ denote the set of components of $C\cap C'$ having at least two nodes, and let $\mathcal{E}$ denote the set of $C$-ears that are subpaths of $C'$. 
Observe that all the paths in these two sets are edge-disjoint, and $C'$ is the union of all these paths. 
Also, all paths in $\mathcal{P}$ are $\Sigma$-even. 
Since $C'$ is $\Sigma$-even, it follows that there is an even number of $\Sigma$-odd paths in $\mathcal{E}$. 
Moreover, at least one of the $C$-ears in $\mathcal{E}$ is parity-breaking, since $C'$ has odd length. 
Consider such an ear, and let $\ell \in \{-1,1\}$ be such that it is contained in $H^{\ell}_C$. 
This ear is $\Sigma$-odd, since $H^{\ell}_C$ has no $2$-sided odd cycle. 
This shows the first part of the statement. 

Next, we observe that there is an odd number of parity-breaking $C$-ears in $\mathcal{E}$, because $C'$ has odd length. 
Since all such ears are $\Sigma$-odd, and since there is an even number of $\Sigma$-odd ears in $\mathcal{E}$, it follows that there is a $\Sigma$-odd {\em non} parity-breaking ear in $\mathcal{E}$. 
By~\eqref{eq:prop1}, such an ear must be contained in $H^{\ell'}_C$ where $\ell' \in \{-1,1\} \setminus \{\ell\}$.  
Property~\eqref{eq:prop2} follows. 

Our goal now is to use~\eqref{eq:prop1} and~\eqref{eq:prop2} to show the existence of more than $k$ pairwise node-disjoint $2$-sided odd cycles in $H$, and thus also in $G$, which will give the desired contradiction. 

Let $b:=b(g,k)$, $c:=c(g,k)=2(b+1)^2+1$ (which is an odd integer), and $m:=(c-1)/2$. 
The set $X_1 \cup X_2$ avoids at least $q - |X_1| - |X_2| \geq c$ cycles among $C_1, \dots, C_q$. 
Consider $c$ such cycles and the cylinder they define. 
Let $J$ be the subgraph of $H$ drawn on that cylinder. 
Observe that $J$ has facewidth at least $t - |X_1| - |X_2| \geq c$, since every noncontractible closed curve in the cylinder is also noncontractible in $\surf$. 
Observe also that every curve from one cuff of the cylinder to the other cuff intersects each of the $c$ cycles we consider. 
By Theorem~\ref{thm:cylinderlinkage}, there are $c$ pairwise node-disjoint cylinder-essential cycles in $J$---let us denote them $D_{-m}, \dots, D_0, \dots, D_m$ in an order consistent with the cylinder in such a way that $C_1$ is contained in $H^{-1}_{D_{-m}}$---and $c$ pairwise node-disjoint $V(D_{-m})$--$V(D_m)$ paths $P_1, \dots, P_c$ such that the intersection of each path $P_i$ with each cycle $D_j$ is a path, which we denote $P^j_i$. 
The reason for using a numbering centered around $0$ for the cycles will become clear in what follows.

For $\ell\in \{-1,1\}$ and $i\in [c]$, let $Q^{\ell}_i$ be the (unique) $V(D_0)$--$V(D_{\ell})$ path contained in $P_i$, and let $u^{\ell}_i$ be the endpoint of $P^{\ell}_i$ incident to $Q^{\ell}_i$. 
For $\ell\in \{-1,1\}$ and $j\in [b]$, let $s^{\ell}_j := u^{\ell}_{j(b+1)}$ and $t^{\ell}_j := u^{\ell}_{2j(b+1)}$. 

Next, we show: 
\begin{equation}
\label{eq:prop3}
\textrm{
Let $\ell\in \{-1,1\}$. 
Then there is a linkage in $H^{\ell}_{D_{\ell}}$ linking $s^{\ell}_j$ to $t^{\ell}_j$ for each $j\in [b]$. 
}
\end{equation}
This can be seen as follows. 
First, by~\eqref{eq:prop2} there is a $1$-sided $D_{\ell}$-ear in $H^{\ell}_{D_{\ell}}$, there is a crosscap in the surface $\surf^{\ell}_{D_{\ell}}$.  
Using this crosscap we see that the linkage is topologically feasible. 
Next we observe that every noncontractible closed curve in the later surface is also noncontractible in $\surf$. 
This implies that the facewidth of $H^{\ell}_{D_{\ell}}$ is at least that of $H$, which is at least $b$, as seen before.  
Hence, by Theorem~\ref{thm:linkage}, the desired linkage exists in $H^{\ell}_{D_{\ell}}$.

For each $\ell\in \{-1,1\}$, consider the linkage given by~\eqref{eq:prop3}, and for each $j\in [b]$ let $E^{\ell}_j$ be the path obtained by taking the union of the path linking $s^{\ell}_j$ to $t^{\ell}_j$ in the linkage with $Q^{\ell}_{j(b+1)}$ and  $Q^{\ell}_{2j(b+1)}$. 
Observe that $E^{\ell}_1, \dots, E^{\ell}_b$ are pairwise node-disjoint $D_0$-ears. 

We prove: 
\begin{equation}
\label{eq:prop4}
\textrm{
Let $\ell\in \{-1,1\}$. 
Then at least $b-3g$ of the paths $E^{\ell}_1, \dots, E^{\ell}_b$ are $1$-sided. 
}
\end{equation}

Choose an arbitrary orientation of the cycle $D_{\ell}$. 
For each $j\in [b]$, consider the cycle $F^{\ell}_j$ obtained by combining the ear $E^{\ell}_j$ with the subpath of $D^{\ell}$ directed from $s^{\ell}_j$ to $t^{\ell}_j$ in our orientation of $D_{\ell}$. 
Observe that the cycle $F^{\ell}_j$ cannot separate the surface $\surf$. 
In particular, $F^{\ell}_1, \dots, F^{\ell}_b$ are all noncontractible cycles. 
Furthermore, if $F^{\ell}_j$ and $F^{\ell}_{j'}$ with $j, j'\in [b], j\neq j'$ are homotopic, then both cycles must be $1$-sided. 
A result of Malni\v{c} and Mohar~\cite[Proposition~3.6]{MM92} implies that $F^{\ell}_1, \dots, F^{\ell}_b$ are partitioned into at most $3g$ equivalence classes by the homotopy relation. 
By the previous remark, it follows that at most $3g$ of the cycles $F^{\ell}_1, \dots, F^{\ell}_b$ are $2$-sided. 
Hence, at least $b-3g$ of these cycles are $1$-sided, and the corresponding paths among $E^{\ell}_1, \dots, E^{\ell}_b$ are $1$-sided as well, showing~\eqref{eq:prop4}. 

Since $b - 6g \geq k+1$, using \eqref{eq:prop4} once for $\ell=-1$ and once for $\ell=1$, we deduce that there exists $I\subseteq [b]$ with $|I|=k+1$ such that for each $j\in I$, both $E^{-1}_j$ and $E^{1}_j$ are $1$-sided. 
For each $j\in I$, let $A_j$ be the cycle obtained by taking the union of $E^{-1}_j$, $E^{1}_j$, $P^0_{j(b+1)}$, and  $P^0_{2j(b+1)}$. 
Note that $A_j$ is $2$-sided, since $P^0_{j(b+1)}$ and $P^0_{2j(b+1)}$ are $2$-sided. 
Now, combining~\eqref{eq:prop1} and~\eqref{eq:prop2}, we deduce that one of the two $D_0$-ears $E^{-1}_j$, $E^{1}_j$ is parity-breaking while the other is not. 
Thus the cycle $A_j$ has odd length. 
Therefore, $A_1, \dots, A_{k+1}$ is a collection of $k+1$ pairwise node-disjoint $2$-sided odd cycles in $H$, giving the desired contradiction. \\

We conclude with a comment about the orientable case. 
To obtain the upper bound stated in the theorem, it suffices to redo the above proof with $f(g,k)=19^{g+1}k$, and use ``$t \leq 8(19^{g} + 19^{2})k + 16$'' as the condition for Case~1. 
In Case~2 we then have $q \geq f(g_1, k) + f(g_2, k) + 1$, and then just after the sets $X_1$ and $X_2$ are defined we quickly conclude the proof as follows. By contradiction, suppose $H=G-(X_1 \cup X_2)$ has a $2$-sided odd cycle. 
Since $q > |X_1| + |X_2|$, one of the cycles $C_1, \dots, C_q$ avoids $X_1 \cup X_2$. 
Using that cycle plus a $2$-sided odd cycle $C'$ of $H$, we follow the arguments of the proof of~\eqref{eq:prop2} and deduce that $\surf$ is in fact nonorientable, a contradiction. 
\end{proof}

\begin{lem} \label{EPodd2walks}
    Assume that $G$ is any graph embedded in a surface $\surf$ in such a way that no odd cycle is $2$-sided. Then there exists $Y \subseteq V(G)$ with $|Y| \leqslant c$ such that $G-Y$ does not contain a $2$-sided odd closed walk, where $c = c(\surf)$ denotes the maximum number of disjoint $1$-sided simple closed curves in $\surf$. Since the Euler genus $g = g(\surf)$ of $\surf$ is at least $c$, $|Y| \leqslant g$.
\end{lem}

\begin{proof}
    We prove the first part of the lemma by induction on the number of nodes. The second part is obvious.

    If $G$ is $2$-connected, then there is nothing to prove since by Lemma~\ref{lem:equivalent_signatures}, $G$ does not contain any $2$-sided odd closed walk. The same argument applies if all connected components of $G$ are $2$-connected.

    Now, assume that $G$ has a component that is not $2$-connected, let $B$ be an endblock of that component, and $v$ be the corresponding cutnode.

    First, assume that $B$ does not contain any $1$-sided cycle. Then, every closed walk in $B$ is $2$-sided. In fact, by Lemma \ref{lem:equivalent_signatures}, every closed walk in $B$ is $2$-sided and even. Thus no minimal $2$-sided odd closed walk in $G$ contains an inner node of $B$. We remove the inner nodes of $B$ from $G$ to obtain a proper subgraph $G'$. By induction, there is a set $Y' \subseteq V(G')$ of size at most $c = c(\surf)$ such that $G' - Y'$ has no $2$-sided odd walk. Then the set $Y := Y'$ is as required.

    Second, assume that $B$ has a $1$-sided cycle. Then $G' := G - v$ can be embedded in a surface $\surf'$ with $c(\surf') \leqslant c(\surf) - 1$. By induction, there is a set $Y' \subseteq V(G')$ of size at most $c(\surf')$ such that $G' - Y'$ has no $2$-sided odd walk. Then we may take $Y := Y' \cup \{v\}$.
\end{proof}

\section{Making the algorithm faster}
\label{sec:FPT}

In this section we discuss how to speed up the two main steps of our algorithm, which are (i) finding a small transversal for $2$-sided odd walks in embedded graphs, (ii) solving the special instance of the minimum cost homologous circulation problem that arises in our context. We explain how to obtain an FPT algorithm for step (i). However, this remains an open question for step (ii).

Finding a transversal for the $2$-sided odd walks in an embedded graph $G$ with signature $\Sigma$ can be done through the \emph{double cover} of $G$, which is defined exactly as the planar double cover of a projective plane graph, see Section~\ref{secEP}.

\begin{obs} \label{obs:RSV}
    Let $G'$ denote the double cover of $G$. A node subset $X \subseteq V(G)$ meets all $2$-sided odd closed walks of $G$ if and only if $X' := \{v^1,v^2 \mid v \in X\}$ meets all odd cycles of $G'$.
\end{obs}

Thanks to Observation~\ref{obs:RSV}, finding a small transversal $X$ for the $2$-sided odd closed walks of $G$ reduces to finding a small odd cycle transversal $X'$ in the double cover $G'$, which can be done with the well-known FPT algorithm of Reed, Smith and Vetta~\cite{RSV04}. If we use the improved algorithm by Lokshtanov, Saurabh and Sikdar~\cite{LSS09}, one can find a suitable transversal $X$ in time $O(2t \cdot 3^{2t} \cdot |V(G)| \cdot |E(G)|)$, where $t := f(g) \cdot k + g$ is an upper bound on $|X|$.

Chambers, Erickson and Nayyeri~\cite{CEN12} give an FPT algorithm for the minimum cost homologous circulation problem in undirected graphs embedded in orientable surfaces, with respect to real homology. It is not clear to us that their algorithm applies to our instances of the minimum cost homologous circulation problem (see Problem \ref{probCirculationHomologous}). The reason is that in our case, homology is over the integers, the surface is non-orientable and the graph is directed. We therefore propose the following as an open problem.  

\begin{conj} \label{conj:FPThomology}
Let $D$ be a directed graph cellularly embedded on a (possibly non-orientable) surface $\surf$ of Euler genus $g$, $c: A(D) \to \mathbb{R}_{\geq 0}$ be a cost function, and $y \in \mathbb Z_{\geq 0}^{A(D)}$ be a circulation.  Then a minimum $c$-cost circulation in $D$ that is homologous to $y$ can be computed using at most $g^{O(g)} n^{3/2}$  arithmetic operations, where $n$ is the number of nodes of $D$.
\end{conj}

\section{An extended formulation}
\label{secEF}

In this section, we establish a proof of Theorem~\ref{thm:extfor}. Let $ G $ be a graph with $ \ocp(G) \le k $ embedded in a surface $ \surf $ with Euler genus at most $ g $. Regarding $ k, g $ as constants, we will show that our approach yields a polynomial-size extended formulation for the stable set polytope $ \stab(G) $ of $ G $. In what follows, we first argue that it suffices to establish the claim for graphs satisfying the standard assumptions. In the second part, we show that it actually suffices to obtain a small extended formulation for $ \sub(G) $ instead of $ \stab(G) $. Finally, in the third part, will derive an extended formulation for $ \sub(G) $ for the case of graphs satisfying the standard assumptions.

\subsection{Restricting to graphs satisfying the standard assumptions}

We know that there is a subset $ X \subseteq V(G) $ with $ |X| \le f(g) \cdot k + g $ such that $ G - X $ satisfies Assumption~\ref{assOddCycles1Sided}. That is, every odd closed walk in $ G - X $ is $ 1 $-sided. For a partition $ X = X_0 \cup X_1 $ consider the polytope
\[
    P_{X_0, X_1} := \{ x \in \stab(G) \mid x(v) = 0 \text{ for all } v \in X_0, \, x(v) = 1 \text{ for all } v \in X_1 \}\,.
\]
Note that $ P_{X_0, X_1} $ is the stable set polytope of a subgraph of $ G - W $ and that
\[
    \stab(G) = \conv \left(\bigcup_{X_0, X_1} P_{X_0, X_1}\right) \,,
\]
where the union is over all partitions of $ X $. Let $ \xc(P) $ denote the smallest size of an extended formulation for $ P $. A result of Balas~\cite{Balas79} states that if $ P $ is the convex hull of the union of polyhedra $ P_1,\dotsc,P_t $ with identical recession cones, then $ \xc(P) \le \sum_{i=1}^t \xc(P_i) + t $. Thus, we may assume that $ G $ satisfies Assumption~\ref{assOddCycles1Sided}.

If $ G $ is not 2-connected, there exist subgraphs $ G_1,\dotsc,G_t $ of $ G $ whose union is $ G $ and where each $ G_i $ is 2-connected and disjoint from all graphs in the list except $ G_{i-1} $ and $ G_{i+1} $ for which we have $ |V(G_i) \cap V(G_{i \pm 1})| \le 1 $. By Chv\'atal's clique cutset lemma~\cite[Theorem 4.1]{Chvatal75}, $ \xc(\stab(G)) \le \sum_{i=1}^t \xc(\stab(G_i)) $. Hence, we may assume that $ G $ is $ 2 $-connected.

We may assume that $ G $ is non-bipartite, otherwise $ \stab(G) = \{ x \in [0,1]^{V(G)} \mid x(v) + x(w) \le 1 \text{ for all } vw \in E(G) \} $ and we are done. By the paragraph below Consequence~\ref{consqNonOrientable} we see that $ \surf $ must be non-orientable. Finally, we may also assume that $ G $ cannot be embedded in a surface of Euler genus $ g - 1 $. Otherwise, we consider such an embedding and repeat the above argumentation. Note that this process will be repeated at most $ g $ times. The paragraph below Consequence~\ref{consqCellular} shows that the embedding of $ G $ is cellular. This means that $ G $ now satisfies the standard assumptions.

\subsection{Restricting to $ \sub(G) $}

Let $ M $ be an edge-node incidence matrix of $ G $ and consider the polyhedron $ \sub(G) = \conv \{ x \in \Z^{V(G)} \mid Mx \le \onevec \} $. To obtain a small extended formulation for $ \stab(G) = \conv \{ x \in \{0,1\}^{V(G)} \mid Mx \le \onevec \} $, the following fact shows that it suffices to find one for $ \sub(G) $.

\begin{prop}
    \label{propStabVsSub}
    For every graph $ G $ we have $ \stab(G) = \sub(G) \cap [0,1]^{V(G)} $.
\end{prop}

Thus, given an extended formulation for $ \sub(G) $ we only need to add at most $ 2|V(G)| $ linear inequalities (describing $ [0,1]^{V(G)}) $ to obtain a description for $ \stab(G) $. Notice that if $ M \in \{0,1\}^{m \times n} $, then the above claim reads
\begin{equation}
    \label{eqSuprisingSetPackingIdentity}
    \conv \{ x \in \{0,1\}^n \mid Mx \le \onevec \} = \conv \{ x \in \Z^n \mid Mx \le \onevec \} \cap [0,1]^n\,.
\end{equation}
This identity appears rather surprising since it is not satisfied by general $ 0/1 $ matrices $ M $: As an example, consider the graph depicted in Figure~\ref{figCounterExample} and let $ M \in \{0,1\}^{16 \times 24} $ be its node-edge (as opposed to edge-node) incidence matrix. On the one hand, there is no $ x \in \{0,1\}^{24} $ with $ Mx = \onevec $ since the graph has no perfect matching. On the other hand, it can be seen that the vector $ \frac{1}{3} \onevec \in \R^{24} $ is contained in the set of the right-hand side of Equation~\eqref{eqSuprisingSetPackingIdentity}. Since the graph is $ 3 $-regular, we have $ M (\frac{1}{3} \onevec) = \onevec $. This shows that the sets in Equation~\eqref{eqSuprisingSetPackingIdentity} are not identical.

\begin{figure}
    \centering
    \includegraphics[scale=0.5]{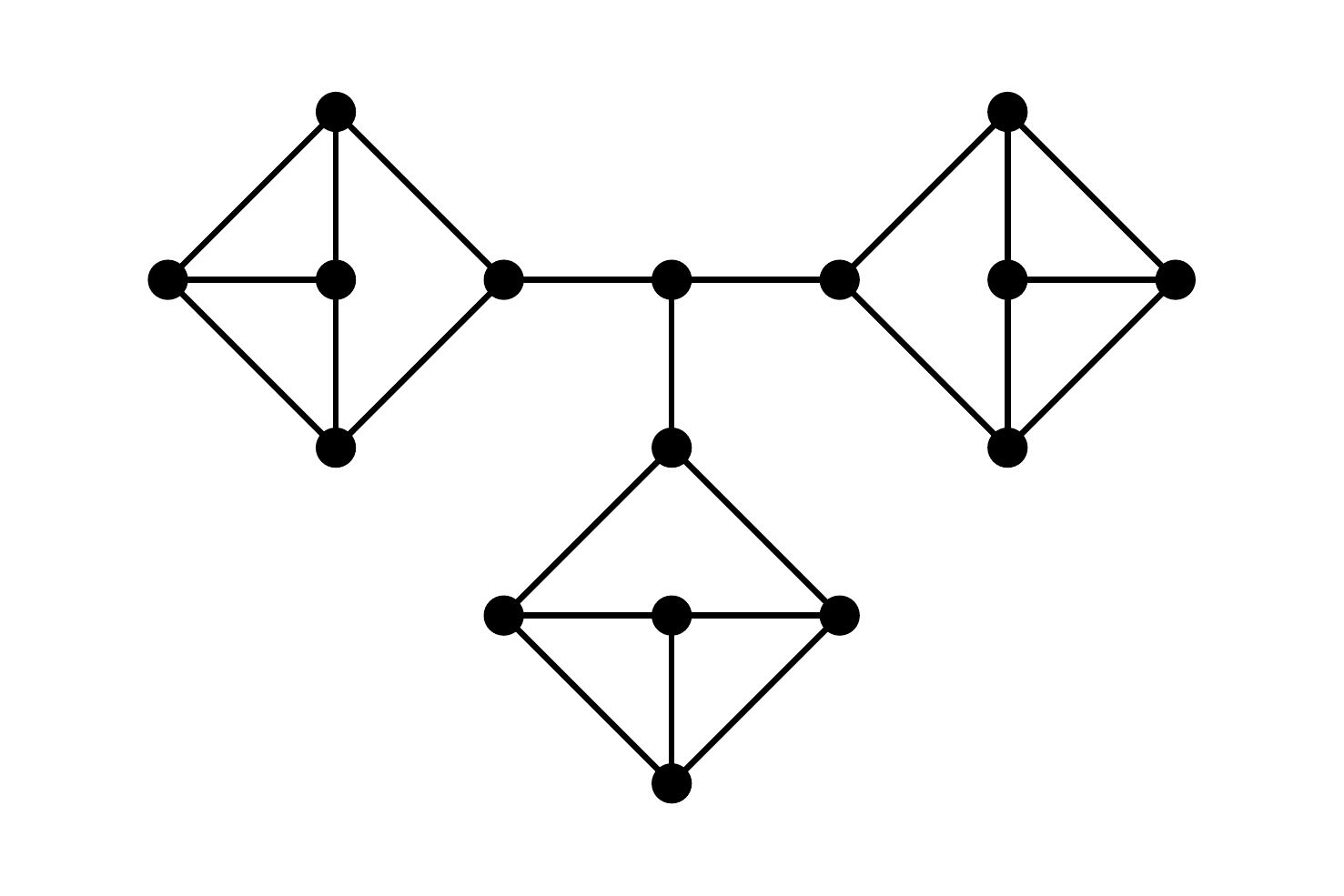}
    \caption{A graph whose \emph{node-edge} incidence matrix does not satisfy Equation~\eqref{eqSuprisingSetPackingIdentity}.}
    \label{figCounterExample}
\end{figure}

In the proof of Proposition~\ref{propStabVsSub} we make use of the following lemma.

\begin{lem}
    \label{lem:P(G)_projection}
    Let $G$ be any graph, and let $v_0$ be any fixed node of $G$. The projection of $\sub(G)$ onto the coordinates indexed by $V(G-v_0)$ equals $\sub(G-v_0)$.
\end{lem}

\begin{proof}
    To see that the projection of $\sub(G)$ is contained in $\sub(G-v_0)$, it suffices to prove that every integer point $x \in \sub(G)$ projects to a point in $\sub(G-v_0)$. Let $x' \in \Z^{V(G-v_0)}$ be the projection of $x$. Then, for every edge $vw$ in $G-v_0$ we have $x'(v) + x'(w) = x(v) + x(w) \leqslant 1$ and hence $x' \in \sub(G-v_0)$, as claimed.

    Conversely, let $x' \in \Z^{V(G-v_0)}$ be any integer point in $\sub(G-v_0)$. Consider a point $x \in \Z^{V(G)}$ that projects to $x'$. By decreasing $x(v_0)$ by a sufficiently large integer amount, we may assume that that $x(v) + x(w) \leqslant 1$ for all edges $vw \in E(G)$. Hence, $x$ is an integer point in $\sub(G)$. We conclude that the projection of $\sub(G)$ contains $\sub(G-v_0)$.
\end{proof}

\begin{proof}[Proof of Proposition~\ref{propStabVsSub}]
    It suffices to show that the polytope $\sub(G) \cap [0,1]^{V(G)}$ is integer. We establish this claim by induction on the number of nodes of $G$. The statement is clearly true if $G$ consists of a single node. Now assume that $G$ has at least two nodes, and the statement holds for all proper induced subgraphs of $G$. We have to show that it holds for $G$ itself.

    We may assume that $G$ is connected. If not, then let $G_1$ and $G_2$ be disjoint and proper induced subgraphs of $G$ whose union is equal to $G$, and in particular $\sub(G) = \sub(G_1) \times \sub(G_2)$. By the induction hypothesis we know that $ \sub(G_1) \cap [0,1]^{V(G_1)} $ and $ \sub(G_2) \cap [0,1]^{V(G_2)} $ are integer and hence $\sub(G) \cap [0,1]^{V(G)} = (\sub(G_1) \cap [0,1]^{V(G_1)}) \times (\sub(G_2) \cap [0,1]^{V(G_2)}) $ is integer as well.

    Now consider any vertex $x^*$ of $ \sub(G) \cap [0,1]^{V(G)} $. Let $V_0 \subseteq V(G) $ denote the set of nodes $v$ such that $x^*(v) = 0$ and $V_1 \subseteq V(G)$ denote the set of nodes $v$ such that $x^*(v) = 1$.

    Let us first consider the case that $V_0 = \emptyset$. We claim that also $V_1 = \emptyset$. Suppose not, so $x^*(v) = 1$ for some $v \in V(G)$. Let $w \in V(G)$ be a neighbor of $v$. Such a node exists since $G$ is connected and has at least two nodes. Since $x^*(w) \geqslant 0 $ and $ x^*(v) + x^*(w) \leqslant 1 $, we obtain $x^*(w) = 0$, a contradiction to $V_0 = \emptyset$. So, in this case we would have $0 < x^*(v) < 1$ for all $v \in V(G)$, implying that $x^*$ is a vertex of $\sub(G)$. However, vertices of $\sub(G)$ are integer and hence we arrive at another contradiction.

    Thus, there must exist a node $v_0 \in V_0$. By Lemma~\ref{lem:\sub(G)_projection}, the projection of $x^*$ onto the coordinates indexed by $V(G-v_0)$ belongs to $\sub(G-v_0) \cap [0,1]^{V(G-v_0)}$. By induction, this projection can be expressed as a convex combination of 0/1-points in $\sub(G-v_0)$. Thus, there exist stable sets $S_1,\dotsc,S_k$ of $G-v_0$ and coefficients $\lambda_1, \dotsc, \lambda_k \in \R_{\ge 0}$ such that $\sum_{i} \lambda_i = 1$ and
    \[
    x^*(v) = \sum_{i : v \in S_i} \lambda_i
    \]
    for all $v \in V(G-v_0)$. Since $x^*(v_0) = 0$, the equation above also holds for $v = v_0$. Now, every stable set of $G-v_0$ is also a stable set of $G$. It follows that $x^*$ is a convex combination of $0/1$-points in $\sub(G)$.
\end{proof}

\subsection{The construction}

It remains to construct a small extended formulation for $ \sub(G) $ in the case that $ G $ satisfies the standard assumptions. Recall that in this case the polyhedron $ \slack(G) = \sigma(\sub(G)) $ is affinely isomorphic to $ \sub(G) $, and hence $ \xc(\slack(G)) = \xc(\sub(G)) $. Thus, it suffices to give a small extended formulation for $ \slack(G) $. To this end, let $ (D,\omega) $ be a dual representation of $ G $ and consider again the cover graph $ \bar{D} $.

For each node $ (f,b) \in V(\bar{D}) $ with $ b \ne \zerovec $ we define the polyhedron $ \bar{Q}_{f,b} $ as the convex hull of non-negative integer unit flows from $ (f,\zerovec) $ to $ (f,b) $ in $ \bar{D} $. It is a well-known that each $ Q_{f,b} $ can be described using linearly many (in the size of $ \bar{D} $) linear inequalities (actually without using additional variables).

Let $ \ell = O(g) $ be the constant in Lemma~\ref{lem:decomposition} and consider the set
\[
    \bar{Q} := \conv \Bigg( \bigcup_{t \in [\ell]} \bigg( \bigcup_{\substack{(f_1,b_1), \dotsc, (f_t,b_t) \in V(\bar{D}) \\ b_1 + \dotsb + b_t = (1,\zerovec) \\ b_1,\dotsc,b_t \ne \zerovec }} \left( \bar{Q}_{f_1,b_1} + \dotsb + \bar{Q}_{f_t,b_t} \right) \bigg) \Bigg) \,.
\]
Notice that the recession cone of each $ \bar{Q}_{f,b} $ is equal to the space of circulations in $ \bar{D} $. This implies that $ \bar{Q} $ is actually a polyhedron with the same recession cone. This also means that we can apply Balas' theorem. Using the simple estimate $ \xc(P_1 + \dotsb + P_t) \le \xc(P_1) + \dotsb + \xc(P_t) $ for Minkowski sums we thus obtain that there exists a polynomial-size (provided that $ g $ is constant) extended formulation for $ \bar{Q} $.

Let $ \pi : \R^{A(\bar{D})} \to \R^{A(D)} $ be the linear projection given in Definition~\ref{defnCoverGraphAlgorithm}. It remains to show that $ \pi(\bar{Q}) = \slack(G) $.

To see that $ \pi(\bar{Q}) \subseteq \slack(G) $ let $ (f_1,b_1), \dotsc, (f_t,b_t) $ be any nodes in $ \bar{D} $ with $ b_1 + \dotsb +  b_t = (1,\zerovec) $, and for each $ i \in [t] $ let $ \bar{y}_i $ be a non-negative integer unit flow from $ (f_i,\zerovec) $ to $ (f_i,b_i) $ in $ \bar{D} $. By Observation~\ref{obsProjection} we have that $ y := \pi(\bar{y}_1 + \dotsb + \bar{y}_t) $ is a non-negative integer circulation in $ D $ such that
\[
    \omega(y) = \omega(\pi(\bar{y}_1)) + \dotsb + \omega(\pi(\bar{y}_t)) = b_1 + \dotsb + b_t = (1,\zerovec)\,.
\]
Thus, $ y \in \slack(G) $, as desired.

Next, let us compare the recession cones of $ \pi(\bar{Q}) $ and $ \slack(G) $. Recall that the recession cone of $ \bar{Q} $ is the space of circulations in $ \bar{D} $. By Observation~\ref{obsProjectionOfCirculation} we obtain that the projection of the recession cone of $ \bar{Q} $ under $ \pi $ is equal to the set of circulations $ y $ in $ D $ with $ \omega(y) = \zerovec $. On the other hand, this is also the recession cone of $ \slack(G) $.

Thus, in order to show $ \slack(G) \subseteq \pi(\bar{Q}) $ it suffices to show that every vertex of $ \slack(G) $ is contained in $ \pi(\bar{Q}) $. To this end, let $ y $ be a vertex of $ \slack(G) $. By Lemma~\ref{lem:decomposition} there exist $ y_1,\dotsc,y_t $ with $ t \in [\ell] $ and $ y = y_1 + \dotsb + y_t $ such that each $ y_i $ is the characteristic vector of a strongly connected (directed) Eulerian subgraph of $ D $ that is non-homologous to $ \zerovec $. By Lemma~\ref{lemLift} we have that for every $ i \in [t] $ there exist $ (f_i, b_i) \in V(\bar{D}) $ and a vector $ \bar{y}_i \in \bar{Q}_{f_i,b_i} $ with $ \pi(\bar{y}_i) = y_i $. Furthermore, by Observation~\ref{obsProjection} we know that
\[
    b_1 + \dotsb + b_t = \omega(y_1) + \dotsb + \omega(y_t) = \omega(y) = (1,\zerovec)\,.
\]
This implies that $ \bar{y}_1 + \dotsb + \bar{y}_t $ is contained in $ \bar{Q} $. Thus, we obtain $ y = \pi(\bar{y}_1 + \dotsb + \bar{y}_t) \in \pi(\bar{Q}) $, as claimed.

\section*{Acknowledgements}

We thank Paul Wollan for insightful discussions regarding the Erd\H{o}s-P\'osa results in Section~\ref{secEP}, and Ahmad Abdi for discussions that led to the example depicted in Figure~\ref{figCounterExample}.  We also thank Erin Chambers for valuable feedback on Conjecture~\ref{conj:FPThomology}.   

Samuel Fiorini and Tony Huynh acknowledge support from ERC grant \emph{FOREFRONT} (grant agreement no. 615640) funded by the European Research Council under the EU's 7th Framework Programme (FP7/2007-2013). 
Gwena\"el Joret acknowledges support from an ARC grant from the Wallonia-Brussels Federation of Belgium. 

\bibliographystyle{abbrv}
\bibliography{references}{}

\end{document}